\documentclass[letter,11pt]{article}
\pdfoutput=1

\usepackage{tcs}
\newcommand{\BSS}{\mathrm{BSS}}

\usepackage{listings}
\title{A Simple Algorithm for Best Separable State}
\author{
Prashanti Anderson\\
MIT\\
\texttt{paanders@mit.edu}
\and
Samuel B. Hopkins\\
MIT\\
\texttt{samhop@mit.edu}
\and
Amit Rajaraman\\
MIT\\
\texttt{amit\_r@mit.edu}
}

\begin{document}

\maketitle

\thispagestyle{empty}

\begin{abstract}
We study the \emph{best separable state} problem (BSS), which asks for the maximum acceptance probability of a quantum measurement over unentangled states.
 In classical terms, the goal is to maximize $\iprod{(x \tensor y), M (x \tensor y)}$ over unit vectors $x,y$ where $0 \preceq M \preceq I$; we call this value $\BSS(M)$.
We study $\BSS$ in the ``perfect completeness'' regime, where given $M$ such that $\BSS(M) = 1$ the goal is to find the best possible solution $x,y$ -- this generalizes the problem of finding a rank-one matrix as close as possible to a given subspace of $\R^{n \times n}$ guaranteed to contain a rank-one matrix.

The strongest known algorithmic guarantees for this problem are: (1) an algorithm which finds a solution with value $1-\epsilon$ in time $\exp(\sqrt{n} (\log n)^{O(1)} / \epsilon^2)$, due to Barak, Kothari, and Steurer, and (2) an algorithm which finds a solution with value $q/n$ in time roughly $n^{O(q)}$, due to Bhattiprolu, Ghosh, Guruswami, Lee, and Tulsiani.

We give a much simpler approach to rounding the SoS relaxation, generalizing the canonical ``global correlation rounding'' technique, and obtain a better running time.
Given $M$ with $\BSS(M) = 1$, our algorithm finds a solution with value $1-\epsilon$ in time $n^{O(\sqrt{n/\epsilon})}$, and a solution of value $q/n$ in time $n^{O(\sqrt q)}$.

Using the same techniques, we prove a new variant of the ``pinning lemma,'' a measure-decomposition theorem widely used in LP/SDP rounding, high-dimensional probability, and statistical physics, which we believe is of independent interest.
\end{abstract}

\newpage

\thispagestyle{empty}

\tableofcontents

\clearpage
\setcounter{page}{1}
\pagestyle{plain}

\section{Introduction}

How can we detect quantum entanglement?
In the \emph{best separable state} problem (henceforth, \emph{BSS}), we are asked to take a description of a quantum measurement device and find the maximum probability that the measurement will accept over all un-entangled (``separable'') quantum states.
More formally, given a Hermitian matrix $M \in \C^{n^2 \times n^2}$ with $0 \preceq M \preceq I$, the goal is to find or approximate
\[
\BSS(M) = \max_{\|x\| = 1, \|y\| = 1} \iprod{x \tensor y, M (x \tensor y)}
\]
where $x,y \in \C^n$ and $\| \cdot \|$ is the Euclidean norm.\footnote{In this paper, we adopt the classical notation of $x \tensor y$ for the tensor product of the state $x$ with the state $y$. In bra-ket notation, one would instead write $\max_{x,y} \iprod{ xy | M |xy }$.}
BSS also amounts to a fundamental classical algorithmic problem: find the maximum of a degree-$4$ polynomial over the high-dimensional unit sphere.
While this problem is likely computationally intractable for worst-case polynomials, BSS restricts to an interesting special case: polynomials whose representations as matrices have bounded eigenvalues.

Algorithms for BSS are equivalent to classical upper bounds on the complexity class QMA[2], the class of computational problems whose solutions can be verified by a verifier with access to two un-entangled quantum provers.\footnote{The choice of $2$ provers is not essential here, as $\mathrm{QMA}[k] = \mathrm{QMA}[2]$ for any constant $k$ \cite{HarrowM13}.}
The relationship of $\textrm{QMA}[2]$ to classical complexity classes is a longstanding mystery in quantum complexity theory: there is a straightforward upper bound $\mathrm{QMA}[2] \subseteq \mathrm{NEXP}$, but $\mathrm{QMA}[2]$ is widely suspected to be contained in smaller complexity classes.
Beyond quantum complexity, BSS is intimately intertwined with several other important questions.
For example, the $2$-to-$4$ norm is exactly BSS on a structured family of measurement operators, and some reductions in the reverse direction are known \cite{HarrowM13, BarakBHKSZ12, BrandaoH15}; the $2$-to-$4$ norm is also closely related to Small-Set Expansion and the Unique Games Conjecture \cite{Khot02, BarakBHKSZ12}.
And in the special case that $M$ is the projector to a subspace of $n \times n$ matrices, BSS becomes the natural question of finding the closest rank-one matrix to the subspace.

It is a folklore result \cite{NatarajanH17} that the top eigenvalue $\lambda_{\max}(M)$ satisfies $\BSS(M) \leq \lambda_{\max}(M) \leq n \cdot \BSS(M)$, and hence gives a polynomial-time $n$-approximation of $\BSS(M)$.
An exponential-time brute-force algorithm based on discretizing the unit sphere solves $\BSS(M)$ exactly (up to numerical errors).
But these algorithms are very far from the best known lower bounds, which say only that a constant-factor approximation in $n^{o(\log n)}$ time would disprove ETH \cite{HarrowM13} and that computing $\BSS(M)$ up to additive $1/\poly(n)$ gaps is NP-hard \cite{BlierT12}.
Efficient algorithms closer to these lower bounds would have significant consequences, the most striking of which is that a quasipolynomial-time constant-factor approximation algorithm would imply that $\textrm{QMA}[2] \subseteq \textrm{EXP}$. 

The first algorithm to go beyond the spectral $n$-approximation and the brute force exponential-time algorithm, due to Barak, Kothari, and Steurer, uses the sum-of-squares (SoS) hierarchy to solve $\BSS$ in the ``perfect completeness'' case in subexponential time \cite{BarakKS17}.
Here ``perfect completeness'' refers to the case that the BSS value equals $1$.

\begin{theorem}[\cite{BarakKS17}]
  \label{thm:bks-bss}
  There is an algorithm which distinguishes between the cases $\BSS(M) = 1$ and $\BSS(M) \le 1-\epsilon$ in time $\exp\left( \sqrt{n} \log(n)^{O(1)}  \cdot \frac 1 {\epsilon^2} \right)$.
\end{theorem}

\cite{BarakKS17} also solves the search version of this problem: given $M$ such that $\BSS(M) = 1$, find $x,y$ such that $\iprod{x \tensor y, M(x \tensor y)} \geq 1-\epsilon$.
To prove \cref{thm:bks-bss}, \cite{BarakKS17} introduces a sophisticated and challenging-to-analyze rounding scheme for the SoS relaxation of $\BSS(M)$. Digesting their algorithm is an ongoing challenge for researchers in quantum information and approximation algorithms.
Our main contribution is a very simple approach to rounding the SoS relaxation for $\BSS(M)$, using \emph{global correlation rounding}\footnote{While our first result for distinguishing $\BSS(M) = 1$ from $\BSS(M) \leq 1-\eps$ directly analyzes the canonical global correlation rounding algorithm, our second result (for distinguishing $\BSS(M) = 1$ and $\BSS(M) = q/n$) also uses a generalization of conditioning in the rounding procedure.} \cite{Montanari08, RaghavendraT12, BarakRS11},
which improves over \cref{thm:bks-bss}.

\begin{theorem}[Main result (see~\cref{thm:asym-bss-close-to-one-main} and~\cref{thm:asym-bss-low-soundness})]
    \label{thm:main-bss}
    There is an algorithm for best separable state with perfect completeness
    having the following guarantees.
    \begin{enumerate}
        \item For every $0<\epsilon<1$, the algorithm distinguishes
        $\BSS(M)=1$ from $\BSS(M)\leq 1-\epsilon$ in time
        $n^{O(\sqrt{n/\epsilon})}$.

        \item For every $1\leq q\leq n/2$, the algorithm distinguishes
        $\BSS(M)=1$ from $\BSS(M)\leq q/n$ in time $n^{O(\sqrt q)}$.
    \end{enumerate}
    Furthermore, in both cases the algorithm solves the corresponding search
    problem in the same running time.
\end{theorem}

In the $1$ versus $1-\epsilon$ regime, our algorithm improves over the state of the art by a $(\log n)^{O(1)} / \epsilon^{1.5}$ factor in the exponent.
This means that our algorithm retains subexponential running time (hence beating brute-force) for any $\epsilon = n^{-1+\Omega(1)}$, whereas the algorithm of \cite{BarakKS17} runs in subexponential time only when $\epsilon = n^{-1/4 + \Omega(1)}$.

In the $1$ versus $q/n$ regime, the fastest previously known algorithm is the general degree-$4$ polynomial optimization algorithm of \cite{BhattiproluGGLT17}, which has running time $n^{O(q)}$ (likely improvable to $2^{O(q)} n^{O(1)}$, per \cite{BhattiproluGGLT17}), while ours has running time $n^{O(\sqrt q)}$.
Thus, our algorithm improves the state of the art for any $q \gg \log^2 n$.

We remark that a simplified proof we present in the technical overview \Cref{sec:tech-overview}, which loses a $\sqrt{\log n}$ factor in the exponent, has a simple analysis spanning only a few pages (!) and improves on the guarantees of \cite{BarakKS17}.

We state \Cref{thm:main-bss} in the setting of optimization over $\C$, but for simplicity in the proofs, we work over $\R$ as the complex version of the problem reduces to our real-valued version (see~\cite[Section B]{BarakKS17}). We also work in the case where $M$ is a projector in the proof sections, as $\BSS$ with perfect completeness reduces to the case of projectors via considering the projector to the top eigenspace.

\paragraph{Open problems}
Many open problems remain, whose solutions would likely introduce new fundamental algorithmic ideas for polynomial optimization.
Two especially interesting examples are:
can the $\sqrt{n}$ in the exponent of the running time can be reduced, perhaps all the way to $O(\log n)$?
And, can $\BSS(M) = 1$ versus $\BSS(M) = 1-n^{-1+\Omega(1)}$ be distinguished in polynomial time?

\paragraph{AI usage statement}
The initial analysis of a pinning-based rounding for best separable state was hatched and executed by the authors without AI assistance, as part of a long-running project studying BSS.
We initially proved a weaker result than \Cref{thm:main-bss}, losing a logarithmic and $1/\sqrt{\epsilon}$ factor in the exponent, and restricted to the ``symmetric'' case, where the goal is to maximize $\iprod{(x \tensor x), M (x \tensor x)}$ rather than $\iprod{(x \tensor y), M (x \tensor y)}$.

AI (GPT 5.5 Pro, GPT 5.6 Sol) was then used to find arguments which remove the logarithmic factor and to generalize the result to the non-symmetric case (involving $\iprod{(x \tensor y), M (x \tensor y)}$).
GPT 5.6 Sol also found the version of the loose-by-$\sqrt{\log n}$ argument we present in the technical overview---our original version, preserved in \Cref{sec:fine-grained-pinning} because of its interesting connection to pinning lemmas, applies only to the symmetric case.
Codex was used to aid in typesetting the paper, filling out related work, and validating citations.

\subsection{Related work}

\paragraph{Sum of Squares Algorithms for Optimizing Constant-Degree Polynomials}
Known algorithmic guarantees for approximating constant-degree polynomial optimization depend strongly on the domain and on the class of polynomials.
For the simplex, de Klerk, Laurent, and Parrilo give a PTAS for fixed-degree polynomial optimization \cite{DeKlerkLP06}.
For the unit sphere, general guarantees about the convergence of the SoS hierarchy are provided by \cite{DohertyW12, FangF21}, but the degree of SoS needed by those works for nontrivial guarantees yield exponential-time algorithms in the dimension.
Worst-case multiplicative approximation for arbitrary even-degree polynomials over the sphere was studied by Bhattiprolu, Ghosh, Guruswami, Lee, and Tulsiani, who give approximation/runtime tradeoffs based on higher levels of SoS \cite{BhattiproluGGLT17}.
Average-case and planted versions of polynomial optimization over the sphere have also been a central testbed for SoS and related spectral or Kikuchi-type algorithms \cite{BhattiproluGL17, HopkinsSSS16, WeinAM25}.
For the especially important $2 \to 4$ norm, a particular class of degree-$4$ polynomials related but not identical to BSS, algorithmic guarantees are known for several families of instances \cite{BarakBHKSZ12, BrandaoH15, HopkinsT26}.

\paragraph{Global Correlation Rounding}
``Global correlation rounding'' refers to the iterated conditioning-based rounding scheme we use here.
It was developed in hierarchy algorithms for constraint satisfaction problems and related SoS rounding frameworks \cite{BarakRS11, RaghavendraT12, BarakKS14, GuruswamiS12}.
Since then, variants of the same idea have been used in scheduling via LP hierarchies \cite{LeveyR21}, mean-field and free-energy approximation for Ising models \cite{JainKR19}, and recent continuous optimization problems such as entrywise low-rank approximation and matrix $p \to q$ norms \cite{AndersonBH26}.
Our use follows this line at the level of the rounding primitive, while the main new ingredient is the fine-grained pinning analysis tailored to the BSS SoS relaxation.

\paragraph{QMA[2] and Best Separable State}
The literature on $\mathrm{QMA}[2]$ is too broad to survey completely here; we highlight the seminal work \cite{HarrowM13} which made the connection between product-state testing, $\mathrm{QMA}[2]$, and polynomial optimization, proving $\mathrm{QMA}[k]=\mathrm{QMA}[2]$ for constant $k$ and, assuming ETH, an $n^{\Omega(\log n)}$ lower bound for constant-gap BSS \cite{HarrowM13}.
The closest prior work to ours is the perfect-completeness SoS algorithm of Barak, Kothari, and Steurer \cite{BarakKS17}; their algorithm follows a different rounding strategy and analysis than ours.

\paragraph{Special Cases of BSS}
Several special cases of BSS admit stronger guarantees than the worst-case problem.
When $M$ is the projector onto a subspace, BSS is the problem of finding the closest rank-one matrix to that subspace; this also captures the perfect-completeness setting.
If the subspace is generic, we know stronger guarantees: Johnston, Lovitz, and Vijayaraghavan give a polynomial-time algorithm for this case \cite{JohnstonLV23}.
Another special case is if $M$ is entrywise nonnegative---then Barak, Kelner, and Steurer give a quasipolynomial-time SoS algorithm \cite{BarakKS14}.

\paragraph{Simultaneous Work} Simultaneous and independent work by Jeronimo, Wu, and Xu \cite{granha2026argmaxprinciplesumofsquaresrelaxations} also obtains an algorithm for $1$ versus $1-\epsilon$ BSS using $O(\sqrt{n/\epsilon})$ degree SoS, with rounding techniques differing substantially from ours.
Another simultaneous and independent work of the same authors \cite{jeronimo2026optimalquantumfinettitheorems} shows that $\BSS(M)$ can be solved up to any constant gap with \emph{im}perfect completeness in time $\exp{(\wt{O}(\sqrt n))}$.

\section{Technical Overview and Warmup}
\label{sec:tech-overview}

In this section we describe a simple algorithm and analysis which achieves the guarantees \Cref{thm:main-bss} up to a $\sqrt{\log n}$ factor in the exponent. We emphasise that these results already improve on the prior known guarantees of \cite{BarakKS17}.

\begin{theorem}[Warmup result]
    \label{thm:bss-tech-overview}
    Suppose $M \in \R^{n^2 \times n^2}$ has $\BSS(M) = 1$.
    Then, for every $0 < \epsilon < 1$, in time $n^{O(\sqrt{n/\epsilon \cdot \log n})}$ we can find unit vectors $x,y \in \R^n$ such that $\iprod{(x \tensor y), M(x \tensor y)} \geq 1-\epsilon$, and for every $1 \leq q \leq n/2$, in time $n^{O(\sqrt{q \log n})}$ we can find unit $x,y \in \R^n$ such that $\iprod{(x \tensor y), M (x \tensor y)} \geq q/n$.
\end{theorem}
At the end of this section we briefly describe how the $\sqrt{\log n}$ factor in the exponent can be removed.

\paragraph{The convex relaxation}
Let $\Pi$ be the projector to the span of eigenvectors having eigenvalue $1$ in $M$.
Since $\BSS(M) = 1$, there are nonzero $x,y$ such that $\Pi(x \tensor y) = x \tensor y$.
We will use an SoS relaxation of the following system of polynomial equations.
(See \cite{barak2016proofs} for background on SoS relaxations.)
\begin{align*}
\mathcal A =\left\{
\begin{gathered}
    \norm{x}_2^2=1,\qquad \norm{y}_2^2=1,\\
    \Pi(x\otimes y)=x\otimes y
\end{gathered}
\right\}.
\end{align*}
By picking a (not too large, $\poly(n)$-size) discrete set $\Sigma$ of $\mathbb{R}$, incurring negligible $1/\poly(n)$-size errors due to discretization, we pass to a discrete assignment problem: find vectors $x,y \in \Sigma^n$ with $\|x\|^2, \|y\|^2 = 1 \pm 1/\poly(n)$ such that $\|\Pi (x \tensor y) \|^2 \geq 1-1/\poly(n)$.
(For the rest of this overview we ignore this negligible $1/\poly(n)$ error.)
This discrete assignment version is a technically convenient departure from the standard SoS relaxation of $\mathcal{A}$ (used by \cite{BarakKS17}), because it allows the \emph{conditioning} operation described below.

Our algorithm first solves the degree $t=$ $O(\sqrt{n/\epsilon \cdot \log n})$ or $t=O(\sqrt{q \log n})$ SoS relaxation of this discrete assignment problem, which can be done in time $|\Sigma|^{O(t)} = n^{O(t)}$.

\subparagraph{SoS background -- pseudodistributions, local distributions, and conditioning}
We defer formal background to \cref{sec:prelims} and discuss here only what is needed for the technical overview.

A solution to this relaxation is a \emph{degree-$t$ pseudodistribution}.
It is often helpful to think of a pseudodistribution as if it were a probability distribution over assignments to $(x,y)$.
We will need two key ways in which pseudodistributions act like such distributions.

First, a degree-$t$ pseudodistribution contains, for every set $S \subseteq [n]$ of size at most $t$, a probability distribution, called a \emph{local distribution}, over assignments in $\Sigma^{|S|}$ to the variables in $S$.
These local distributions are are consistent on their intersections, in that the local distributions on $S, S'$ must have the same marginal on $S \cap S'$.

Second, we can \emph{condition} a pseudodistribution on the value of any single coordinate.
For any variable $x_i$ (or $y_i$), we can draw a sample $\hat{x}_i$ from the local distribution on $\{i\}$ and replace the local distribution on $S$ with the local distribution on $S \cup \{i \}$, conditioned on on the event that $x_i = \hat{x}_i$.
This yields a degree $t-1$ pseudodistribution. 

Finally, some notation: for a function $f(x,y)$ depending on at most $t$ variables, we write $\pE f(x,y)$ for the expected value of $f$ under the appropriate local distribution; we extend the operator $\pE$ linearly to all $f$ which can be written as linear combinations of such local functions.
Thus, for example, $\pE x \in \R^n$ is the vector whose entries are $\E x_i$ with expectation taken under the $\{i\}$-th local distribution.

\subsection{\texorpdfstring{$1$ versus $1-\epsilon$}{1 versus 1-ε}}

We describe and analyze a simple rounding procedure which takes a degree $O(\sqrt{n/\eps \cdot \log n })$ pseudodistribution and outputs vectors $\hat{x},\hat{y}$ such that $\norm{\Pi(\hat{x} \otimes \hat{y})}_2 \geq (1-\eps)\cdot \norm{\hat{x} \otimes \hat{y}}_2$.
In the next section we describe how to adapt the analysis to the $1$ versus $q/n$ setting.

The rounding procedure is simple: enumerate over all possible conditionings of $O\left(\sqrt{n/\eps \cdot \log n}\right)$ coordinates, for each conditioning compute the conditional means $\pE x, \pE y$, and output the best solution found in this manner.
The step where the algorithm outputs $\pE x, \pE y$ succeeds if any of the conditioned pseudoexpectations satisfy the following \emph{small covariance} condition:

\begin{lemma}[Small Covariance Success Condition (Informal)]
    \label{lem:mean-rounding-tech-overview}
    Let $\pE $ be a pseudodistribution satisfying $\mathcal A$.
    If
    \begin{align}
            \norm{\widetilde{\Cov}(x,y)}_F \leq \sqrt \eps \cdot \norm{\pE x}\norm{\pE y} \label{eq:mean-rounding-condition}
    \end{align}
    then
    \[ \norm{\Pi((\pE x) \otimes (\pE y))}_2 \geq (1-\eps) \cdot \norm{(\pE x) \otimes (\pE y)}\,.\] 
\end{lemma}

Here $\widetilde{\Cov}(x,y) = \pE xy^\top - (\pE x)(\pE y)^\top$ and $\| \cdot \|_F$ is the Frobenius norm.
The proof of this lemma is straightforward and follows similar ideas to \cite{BarakKS17}; the main idea is that the upper bound on $\|\widetilde{\Cov}\|_F$ ensures that $\pE x \tensor y$ is close to $(\pE x) \tensor (\pE y)$, and that $\pE x \tensor y$, by virtue of the constraints of the SoS relaxation, must lie in the subspace $\Pi$ projects to.
(This last point is where $\BSS(M) = 1$ is crucial; even if $\BSS(M) = 1-\eps/2$, $\pE x \tensor y$ might actually be far from the subspace.)
Together, these imply $(\pE x)(\pE y)^\top$ must lie close to that subspace.
We defer the formal proof to~\cref{sec:bss-asym} and instead focus on the meat of the argument: that the ``small covariance success condition'' \eqref{eq:mean-rounding-condition} holds after conditioning on the values of at most $O(\sqrt{n/\epsilon \cdot \log n})$ coordinates.

\paragraph{A new pinning lemma}
The technical heart of our result is a \emph{pinning lemma}, which shows that we can establish the condition in~\cref{eq:mean-rounding-condition} via conditioning on (``pinning'') $O\left(\sqrt{n/\eps \cdot \log n}\right)$ coordinates.
\begin{lemma}[Warmup pinning]
\label{lem:pinning-tech-overview}
    Let $0<\eps\leq1$, and let $(x,y)$ be jointly-distributed random vectors supported on
    $\mathbb S^{n-1}\times\mathbb S^{n-1}$ with finite support.
    There exists a set $S \subseteq [2n]$ of at most $|S| \leq O(\sqrt{n/\epsilon \log n})$ coordinates and values $(\hat{x},\hat{y})_S \in \R^{|S|}$ for those coordinates such that $(\hat{x},\hat{y})_S$ is in the support of $(x,y)$ restricted to $S$, and such that 
    \[
        \norm{\Cov(x,y) \, | \, (x,y)_S = (\hat{x},\hat{y})_S}_F
        \leq
        \sqrt{\epsilon} \cdot \norm{\E x \, | \, (x,y)_S = (\hat{x},\hat{y})_S} \cdot \norm{\E y \, | \, (x,y)_S = (\hat{x},\hat{y})_S}\,.
    \]
    The lemma also applies to pseudodistributions of degree at least $O(\sqrt {n/\eps \cdot \log n})$.
\end{lemma}

We now outline the proof of~\cref{lem:pinning-tech-overview}.
We describe a simple algorithm which chooses coordinates $i \in [2n]$ and values $\hat{x}_i$ or $\hat{y}_i$ iteratively, at each step adding one coordinate to the conditioned set $S$.
This proceeds in two phases: after the first phase, both $\|\E x\|^2$ and $\|\E y\|^2$ will be at least $\sqrt{\log n / (\epsilon n)}$.
In the second phase will use these ``floors'' on the distances of the means from the origin to establish~\cref{eq:mean-rounding-condition}.

\subparagraph{First phase}
For the first phase, observe that via greedily choosing coordinates $i$ and values $\hat{x}_i$, we can ensure that the contribution to $\norm{\E x}^2$ from pinned coordinates is at least $r/n$ by pinning $r$ coordinates; the same holds for $y$.
Illustrating with the case $r=1$, since $\sum_{i \leq n} \E x_i^2 = 1$, there is some $i$ such that $\E x_i^2 \geq 1/n$, and hence some value $\hat{x}_i$ in the support of $x_i$ such that $\hat{x}_i^2 \geq 1/n$; we can pin $x_i = \hat{x}_i$ for this choice of $\hat{x}_i$.
The argument for general $r > 0$ can be proved by induction.
Importantly for what follows, the contribution of pinned coordinates to $\| \pE x\|^2$ cannot decrease under subsequent pinning operations, so we can apply this greedy pinning operation to $x$ and then $y$ in sequence.

\subparagraph{Second phase}
We turn to the second phase.
We show that if $\pE$ has the property that pinned coordinates contribute at least $\sqrt{\log n /(\epsilon n)}$ to both $\| \pE x\|^2, \| \pE y\|^2$, then after pinning at most $O(\sqrt{ n \log n  / \epsilon })$ additional coordinates, \cref{eq:mean-rounding-condition} holds.
We show this by analyzing the potential function
\[ \Phi(\pE) = \norm{\pE x}^2 + \norm{\pE y}^2 \,.\]
We will show that whenever~\cref{eq:mean-rounding-condition} does not hold that we can ``make progress'' on the goal via increasing $\Phi(\pE)$ by conditioning on a coordinate. Note that $0 \leq \Phi(\pE) \leq 2$ via our norm constraints on $x,y$ and thus $\Phi(\pE)$ cannot increase indefinitely.
At this level the argument is inspired by the analysis of \emph{global correlation rounding} for CSPs from \cite{BarakRS11}, but the ``success condition'' we are targeting is significantly stronger than what is needed in the CSP setting.

The main idea is that if the covariance $\| \Cov(x,y)\|_F$ is large, then there is some coordinate $x_i$ such that conditioning on $x_i$ significantly reduces the variance of $y$, which in turn increases $\| \pE y \|^2$ and hence increases $\Phi(\pE)$ (or, respectively, a coordinate $y_i$ which we can condition on to reduce the variance of $x$).
The quantitative relationship between covariance of two random variables $a,b$ and the reduction in variance in $a$ given by conditioning on $b$ is captured by the following elementary inequality \cite{BarakRS11}:
\begin{equation}
     \label{eq:variance-reduction-tech-overview}
    \Var(a) - \E_b \Var(a \vert b) \geq \Omega\left(\frac{\Cov^2(a,b)}{\Var(b)}\right) \,.
\end{equation}
Since $x$ is supported on unit vectors, $\norm{\E x}_2^2 = 1 - \tr(\Cov(x))$.
Thus, if we condition on a fixed $x_j$ (with value drawn from its marginal distribution) and let the conditioned pseudodistribution be $\pE'$, then we have that
\begin{align*}
    \Phi(\pE') - \Phi(\pE) &= \sum_{i \in [n]} \left( \Var(x_i) - \E_{x_j} \Var(x_i \vert x_j)\right)+ \sum_{i \in [n]} \left(\Var(y_i) - \E_{x_j} \Var(y_i \vert x_j) \right) \\
    &\geq \frac{\sum_{i \in [n]} \Cov^2(x_i, x_j) + \sum_{i \in [n]} \Cov^2(y_i, x_j)}{\Var(x_j)}\,.
\end{align*}
The same holds for $y_j$ by symmetry, and thus if we pick a random $x_j$ or $y_j$ with probabilities proportional to $\Var(x_y),\Var(y_j)$, via an averaging argument we have
\begin{equation}
    \E \left[\Phi(\pE') - \Phi(\pE)\right] \geq \frac{\norm{\Cov(x,x)}_F^2 + 2\norm{\Cov(x,y)}_F^2 + \norm{\Cov(y,y)}_F^2}{\Tr(\Cov(x)) + \Tr(\Cov(y))} \,. \label{eq:tech-overview-cov-progress}
\end{equation}
Thus there exists some coordinate (and choice of pinning for that coordinate) such that that gain in the potential is at least $\norm{\Cov(x,y)}_F^2$, using non-negativity of norms and the fact that the denominator is at most $2$.
Now, if~\cref{eq:mean-rounding-condition} fails for $\pE$, we have that for this choice of pinning,
\[ \Phi(\pE') - \Phi(\pE) \geq \eps \cdot \norm{\pE x}_2^2 \cdot \norm{\pE y}_2^2 \,.\]
We now aim to relate this lower bound to the current value of the potential $\Phi(\pE)$.
Using the lower bound we established in the first phase on $\| \pE x\|^2, \| \pE y\|^2$, we have that 
\[ \epsilon \cdot \norm{\pE x}_2^2 \norm{\pE y}_2^2 \geq \frac 12 \cdot \sqrt{\frac {\epsilon \log n}{n}} \cdot \left( \norm{\pE x}_2^2 + \norm{\pE y}_2^2\right) = \frac 1 2 \cdot \sqrt{\frac {\epsilon \log n}{n}} \cdot \Phi(\pE)\,.\]
Thus, whenever~\cref{eq:mean-rounding-condition} fails, we can increase the potential function by a \emph{multiplicative} factor of $1 + \Omega(\sqrt{(\eps \log n)/n}) \approx \exp(\Omega(\sqrt{(\epsilon \log n)/n}))$.
At the start of the second phase, $\Phi(\pE) \geq \sqrt{1/(\epsilon n)}$, and it cannot ever become larger than $2$.
So after at most $\sqrt{n \log n/\epsilon}$ pinning operation, this second phase must encounter $\pE$ which satisfies~\cref{eq:mean-rounding-condition}.

\subsection{$1$ versus $q/n$}
Next we describe how to adapt the foregoing analysis to round $\pE$ to $\hat{x},\hat{y}$ satisfying $\| \Pi(\hat{x} \tensor \hat{y})\|^2 \geq (q/n) \cdot \|\hat{x} \tensor \hat{y}\|^2$ using $O(\sqrt{q \log n})$ pinnings.

\paragraph{Rounding via singular vectors}
In this regime, rather than using the means $\pE x, \pE y$, after pinning we output the top left and right singular vectors of $R=\pE xy^\top$.
The constraints imply $\pE(x\tensor y)$ lies in the range of $\Pi$.
This gives the following analogue of \Cref{lem:mean-rounding-tech-overview}.

\begin{lemma}[Singular-Vector Rounding Success Condition (Informal)]
\label{lem:singular-rounding-tech-overview}
    Let $\pE$ be a pseudodistribution satisfying $\mathcal A$, and suppose $R=\pE xy^\top$ is nonzero. If $u,v$ are top left and right singular vectors of $R$, then
    \[
        \norm{\Pi(u\tensor v)}_2^2
        \geq
        \frac{\sigma_1(R)^2}{\norm R_F^2}.
    \]
\end{lemma}
\begin{proof}
Since $\pE x \tensor y$ lies in the range of $\Pi$,
\[
    \sigma_1(R)
    =\iprod{u\tensor v,\operatorname{vec}(R)}
    =\iprod{\Pi(u\tensor v),\operatorname{vec}(R)}
    \leq
    \norm{\Pi(u\tensor v)}_2\norm R_F.
\]
Squaring proves the lemma. 
\end{proof}

\paragraph{A low-soundness pinning lemma}
The analogue of~\cref{lem:pinning-tech-overview} is the following.

\begin{lemma}[Warmup low-soundness pinning]
\label{lem:low-soundness-pinning-tech-overview}
    Let $1\leq q\leq n/2$, and let $(x,y)$ be jointly-distributed random vectors supported on $\mathbb S^{n-1}\times\mathbb S^{n-1}$ with finite support. There is a set $S\subseteq[2n]$ of at most $O(\sqrt{q\log n})$ coordinates and values $(\hat{x},\hat{y})_S$ in the support of $(x,y)_S$ such that, under the resulting conditional distribution,
    \[
         {\norm{\E[xy^\top\mid (x,y)_S=(\hat{x},\hat{y})_S]}_F} \leq \sqrt{\frac n q} \cdot {\sigma_1\bigl(\E[xy^\top\mid (x,y)_S=(\hat{x},\hat{y})_S]\bigr)} 
    \]
    The lemma also applies to pseudodistributions of degree at least $O(\sqrt{q\log n})$.
\end{lemma}
\begin{proof}[Proof sketch]
As in the $1$ versus $1-\epsilon$ case, the argument has two phases.
The first is the same as before: pin $O(\sqrt{q \log n})$ coordinates of each of $x$ and $y$ so that pinned coordinates contribute at least $\sqrt{q \log n} / n$ to each of $\| \pE x \|^2$ and $\| \pE y\|^2$.

Suppose that for some conditioned $\pE$,
\cref{lem:singular-rounding-tech-overview} fails, i.e., $\sigma_1(\pE xy^\top)^2 \leq \frac qn\norm{\pE xy^\top}_F^2$.
We will show that some further pinning increases the same potential function $\Phi(\pE) = \norm{\pE x}_2^2+\norm{\pE y}_2^2$ as before.

The rows of $\pE xy^\top$ corresponding to the pinned coordinates of $x$
equal the vector of pinned $x$ values times $(\pE y)^\top$.
Since the pinned
values contribute at least $\sqrt{q\log n}/n$ to $\norm{\pE x}_2^2$,
\[
    \sigma_1(\pE xy^\top)^2 \geq \frac{\sqrt{q\log n}}{n}\norm{\pE y}_2^2 \, ,
\]
by looking just at the submatrix of rows corresponding to pinned $x$ coordinates.
Symmetrically, $\sigma_1(\pE xy^\top)^2 \geq \frac{\sqrt{q\log n}}{n}\norm{\pE x}_2^2$.
Taking the average,
\[
    \sigma_1(\pE xy^\top)^2 \geq \frac{\sqrt{q\log n}}{2n}\left(\norm{\pE x}_2^2+\norm{\pE y}_2^2 \right)
    =
    \frac{\sqrt{q\log n}}{2n}\Phi(\pE).
\]

Failure of singular-vector rounding now gives
\[
    \norm{\pE xy^\top}_F^2
    \geq
    \frac nq\,\sigma_1(\pE xy^\top)^2\\
    \geq
    \frac12\sqrt{\frac{\log n}{q}}\,
    \Phi(\pE) \, .
\]
It remains to turn this lower bound into a lower bound on $\| \Cov(x,y)\|_F^2$, after which we can apply the same strategy as in the $1$ versus $1-\epsilon$ case (see \eqref{eq:variance-reduction-tech-overview}).
Since $(\pE x)(\pE y)^\top$ has rank one, subtracting it from $\pE xy^\top$ reduces the Frobenius norm by no more than subtracting the top singular vectors would:
\[
    \norm{\Cov(x,y)}_F^2
    =
    \norm{
        \pE xy^\top-(\pE x)(\pE y)^\top
    }_F^2
    \geq
    \norm{\pE xy^\top}_F^2
    -
    \sigma_1(\pE xy^\top)^2 \geq \left(1-\frac qn\right)
    \norm{\pE xy^\top}_F^2 \, ,
\]
where the last step again uses that singular-vector rounding fails so $\sigma_1(\pE xy^\top)^2 \leq (q/n) \| \pE xy^\top \|_F^2$.
Putting this together with our lower bound on $\| \pE xy^\top \|_F^2$,
\[
    \norm{\Cov(x,y)}_F^2
    \geq
    \frac14\sqrt{\frac{\log n}{q}}\,
    \Phi(\pE) \, .
\]
The same variance-decrease calculation as in the $1$ versus
$1-\epsilon$ case now shows that some coordinate and some value for that
coordinate increase $\Phi$ by at least $\norm{\Cov(x,y)}_F^2$. Therefore,
whenever singular-vector rounding fails, some pinning satisfies
\[
    \Phi(\pE')
    \geq
    \left(
        1+\Omega\left(\sqrt{\frac{\log n}{q}}\right)
    \right)
    \Phi(\pE).
\]
This multiplicative growth can therefore occur for at most
\[
    O\left(
        \sqrt{\frac q{\log n}}
        \log\frac{n}{\sqrt{q\log n}}
    \right)
    \leq
    O(\sqrt{q\log n})
\]
additional pinnings in the second phase.
Together with the first phase, this proves \cref{lem:low-soundness-pinning-tech-overview}.
\end{proof}

\subsection{Removing the $\sqrt{ \log n}$ factors}
We briefly outline how the $\sqrt{\log n}$ factors are removed.

\paragraph{1 versus $1-\epsilon$ case}
It turns out that \Cref{lem:pinning-tech-overview} can be directly improved to remove the $\sqrt{\log n}$ factor, which leads immediately to an algorithm with running time $n^{O(\sqrt{n / \epsilon})}$.
To obtain this improvement, we use a new potential function $\Phi'(\pE) = \| \pE x\|^2 + \| \pE y \|^2 + ( \| \pE x\| + \| \pE y\|)^2$.
The basic intuition behind this potential function is to improve on the roughly $(1 + \min(\| \pE x\|^2, \| \pE y\|^2))$ rate of multiplicative growth in the potential function $\Phi$ we used previously.
Note that if $\| \pE x\|^2 \approx \| \pE y\|^2$, then $\Phi$ actually increases much faster than the analysis we gave suggests.
Otherwise, roughly speaking, if $\| \pE x\|^2 \ll \| \pE y\|^2$ (or vice versa), the term $(\| \pE x\| + \| \pE y\|)^2$ rewards upward movement in $\min(\| \pE x\|^2, \| \pE y\|^2)$.
The details are in \Cref{sec:eps-pinning-lemma}.

\paragraph{$1$ versus $q/n$ case}
Removing the $\sqrt{\log n}$ factor in the $1$ versus $q/n$ case again uses the potential function $\Phi'(\pE)$ from the $1$ versus $1-\epsilon$ case, with two additional twists.
First, rather than always output the top singular vectors of $\pE xy^\top$, two additional possible outputs are considered -- one is to output the pair $(\pE x) / \| \pE x\|, v$,  where $v$ is the maximum eigenvector of the quadratic form $v \mapsto \iprod{ (x \tensor v), \Pi (x \tensor v)}$, and the other the same with $x$ and $y$ swapped.
The second twist is that in addition to the coordinate-pinning operation, the algorithm uses a reweighing operation on the pseudodistribution, replacing $\pE p(x,y)$ with $\pE f(x,y)^2 p(x,y) / \pE f(x,y)^2$ for a well-chosen degree-$2$ function $f$.

\subsection{Fine-grained pinning lemma}
We conclude this technical overview with a brief discussion of pinning lemmas, and state one new pinning lemma which goes slightly beyond \Cref{lem:pinning-tech-overview,lem:low-soundness-pinning-tech-overview}.
This section is not essential to the proofs of our main results (and can be safely skipped on first reading), but since pinning lemmas play a crucial role in rounding algorithms for SoS, as well as in statistical physics and high-dimensional probability, the new pinning lemma may be of independent interest.

\paragraph{Background on pinning lemmas}
Pinning lemmas are measure decomposition theorems which show that high-dimensional probability distributions can be written as a mixture of ``few'' distributions, most or some of which are ``simple,'' with various definitions of ``few'' and ``simple,'' by conditioning on, or ``pinning'', a small number of coordinates.

The original pinning lemma has been widely used across algorithms \cite{RaghavendraT12, BarakKS17, AndersonBH26, GuruswamiS12, BarakRS11}, high-dimensional probability, and statistical physics \cite{Montanari08, CojaOghlanKPZ17, JainKR19}.
It says that every distribution becomes approximately pairwise independent after conditioning on a small number of coordinates.
For simplicity in this discussion, we stick to mean-zero distributions on the unit sphere.

\begin{theorem}[Pinning Lemma (variance version), \cite{BarakRS11}]
\label{thm:pinning-variance}
For any distribution $\mu$ on the $n$-dimensional unit sphere and every $\epsilon > 0$ there is a subset of at most $t \leq O(1/\epsilon)$ coordinates $S \subseteq [n]$ such that $\E_{x_S \sim \mu_S} \| \Cov(\mu \, | \, x_S) \|_F^2 \leq \epsilon$.
\end{theorem}

Here, $\mu_S$ denotes $\mu$ restricted to the coordinates in $S$, $\mu \, | \, x_S$ denotes the conditional distribution of $\mu$ after fixing the coordinates in $S$ to the values in $x_S$, and $\| \cdot \|_F$ is the Frobenius norm, so that $\| \Cov(\mu \, | \, x_S) \|_F^2 = \sum_{i,j \leq n} \Cov(x_i, x_j \, | \, x_S)^2$ gives the average squared covariance across pairs of coordinates after conditioning.
Notice that the maximum possible value of $\|\Cov\|_F^2$ for a distribution on the unit sphere is $1$.
If the coordinates of $\mu$ are supported in a discrete set of size $q$ (e.g. $q = 2$ for the hypercube), then at most $q^{O(1/\epsilon)}$ different distributions can appear in the decomposition $\mu = \E_{x_S \sim \mu_S} (\mu \, | \, x_S)$; these conditional distributions give us the interpretation of the pinning lemma as decomposing $\mu$ into a mixture of not-too-many distributions.

\paragraph{A fine-grained pinning lemma}
Like \Cref{thm:pinning-variance}, \Cref{lem:pinning-tech-overview,lem:low-soundness-pinning-tech-overview} also show that conditioning on the values of a few coordinates of a high-dimensional probability distribution results in small covariance.
They differ from \Cref{thm:pinning-variance} in two important ways:
\begin{enumerate}
    \item They measure ``simplicity'' of a distribution not just in terms of $\| \Cov \|_F^2$, but in terms of the ratio between $\| \Cov \|_F^2$ and $\| \E x \|^2, \| \E y\|^2$, meaning that a distribution does not count as ``simple'' unless its covariance is small \emph{compared to its mean}, which is crucial for the application to BSS.
    \item To obtain this stronger notion of simplicity, they make two important sacrifices compared to \Cref{thm:pinning-variance}:
        \begin{itemize}
            \item They require conditioning on many more coordinates, and
            \item They don't offer a true measure decomposition, because the \emph{values} for the coordinates that they condition on are chosen greedily rather than being sampled from the marginal distributions of those coordinates.
        \end{itemize}
\end{enumerate}

While having a measure decomposition as opposed to just some set of values to condition on isn't important for our applications to BSS, this property has been very important in other uses of pinning-style lemmas, especially in the theory of localization schemes \cite{ChenE25}.
This motivates the question of how much of the guarantees of \Cref{lem:pinning-tech-overview,lem:low-soundness-pinning-tech-overview} can be obtained if we insist on a measure decomposition?

It turns out that via more or less the same techniques as we use to prove \Cref{lem:pinning-tech-overview,lem:low-soundness-pinning-tech-overview}, one can prove the following fine-grained pinning lemma, which does provide a measure decomposition.
The expense, compared to \Cref{lem:pinning-tech-overview,lem:low-soundness-pinning-tech-overview}, is that it applies only to a single random vector $x$ rather than a pair $(x,y)$ -- this is the reason we cannot use it to prove our main results about BSS.

\begin{theorem}[Fine-grained pinning lemma, see \Cref{sec:fine-grained-pinning}]
    \label{thm:pinning-fine-grained}
    For every mean-zero distribution $\mu$ on the $n$-dimensional unit sphere and every $0<\epsilon<1$, there is a set of $(S,x_S)$ where $S \subseteq [n]$ and $x_S \in \R^{|S|}$ with $|S| \leq O\left(\frac{\sqrt{n}}{\epsilon}\right)$,
    such that $\E_{(S,x_S)} (\mu \, | \, x_S) = \mu$, and
    \[
    \E_{(S,x_S)} \Norm{ \Cov(\mu \, | \, x_S) }_F \leq O(\epsilon) \cdot \E_{(S,x_S)} \Norm{ \E_{x \sim \mu \, | \, x_S} x }^2 \, .
    \]
\end{theorem}

See \Cref{sec:fine-grained-pinning} for the variant of \Cref{thm:pinning-fine-grained} in the $\eps > 1$ case.

The statement $\E_{(S,x_S)}(\mu \, | \, x_S) = \mu$ captures the ``measure decomposition'' aspect of \cref{thm:pinning-fine-grained}.
It captures that, in the proof of \Cref{thm:pinning-fine-grained}, $(S,x_S)$ is obtained by iteratively choosing one coordinate $i$ at a time, sampling a value $x_i$ from the marginal distribution on that coordinate, and conditioning on that value.

To understand \cref{thm:pinning-fine-grained}, it helps to unpack the different ways in which the condition in~\cref{thm:pinning-variance} can be satisfied.
At a high level, the Frobenius norm of the conditional covariance, $\| \Cov(\mu \, | \, x_S)\|_F$, could be small for two different reasons.
One is that most coordinates have become nearly fixed around their means and hence have small variance, and consequently the mean of the conditioned distribution is noticeably far from the origin.
The other is that a typical pair of coordinates has become nearly independent, but each individual coordinate still has mean close to zero and retains large variance.
\cref{thm:pinning-fine-grained} tells us that, at least in an in-expectation sense, the latter does not always happen, $\| \E_{x \sim \mu \, | \, x_S} x \|^2$ is lower-bounded, so the average per-coordinate variance cannot be too large.
In other words, in expectation the mean drifts noticeably far from the origin.

The cost of this refined guarantee, however, is that we must condition on many more coordinates than in \cref{thm:pinning-variance}, as we must condition on a set of size $\sqrt{n}/\eps$ rather than $1/\eps$.
This factor is tight, as can be shown by the simple example of the uniform distribution on $\left\{ \pm \frac{1}{\sqrt{n}} \right\}^n$.

\begin{figure}[H]
\centering
\begingroup
\newcommand{\CovariancePane}[7]{%
  \begin{scope}[shift={(#1,#2)}]
    \draw[plotbox] (-1.18,-1.02) rectangle (1.18,1.02);
    \draw[axis] (-1.08,0) -- (1.08,0);
    \draw[axis] (0,-.92) -- (0,.92);
    \draw[origin] (-.045,0) -- (.045,0);
    \draw[origin] (0,-.045) -- (0,.045);
    \begin{scope}[shift={(#6,#7)}, rotate=#5]
      \foreach \s/\op in {.35/.45,.62/.65,.89/.85}{%
        \pgfmathsetmacro{\rx}{\s*#3}
        \pgfmathsetmacro{\ry}{\s*#4}
        \draw[contour, opacity=\op] (0,0) ellipse [x radius=\rx cm, y radius=\ry cm];
      }
    \end{scope}
    \fill[mean] (#6,#7) circle (1.4pt);
  \end{scope}
}
\resizebox{\textwidth}{!}{%
\begin{tikzpicture}[
  plotbox/.style={draw=black!30, line width=.35pt},
  axis/.style={draw=black!18, line width=.25pt},
  origin/.style={draw=black!35, line width=.35pt},
  contour/.style={draw=citeblue!75!black, line width=.7pt},
  mean/.style={fill=diffcolor!90!black},
  row-label/.style={font=\small, align=right, anchor=east, text=black!85}
]
\node[font=\small, text=black!70] at (4.5,1.62) {more conditioned coordinates};
\draw[->, draw=black!45, line width=.45pt] (1.1,1.42) -- (7.9,1.42);
\foreach \x/\label in {0/$t=0$,3/$t=1$,6/$t=2$,9/$t=3$} {
  \node[font=\scriptsize, text=black!65] at (\x,1.18) {\label};
}
\node[row-label] at (-1.55,0) {mean stays near $0$};
\node[row-label] at (-1.55,-2.65) {mean drifts};
\CovariancePane{0}{0}{1.02}{.42}{25}{0}{0}
\CovariancePane{3}{0}{.96}{.52}{18}{0}{0}
\CovariancePane{6}{0}{.88}{.66}{8}{0}{0}
\CovariancePane{9}{0}{.70}{.70}{0}{0}{0}
\CovariancePane{0}{-2.65}{1.02}{.42}{25}{0}{0}
\CovariancePane{3}{-2.65}{.96}{.39}{25}{.18}{.08}
\CovariancePane{6}{-2.65}{.90}{.36}{25}{.42}{.20}
\CovariancePane{9}{-2.65}{.84}{.33}{25}{.38}{.28}
\end{tikzpicture}%
}
\endgroup
\caption{Two ways that $\|\Cov\|_F$ can decrease as we condition on coordinates in a high-dimensional distribution. In the top row, correlations disappear while the conditional mean remains near the origin --  the distribution becomes isotropic but retains significant variance. In the bottom row, the overall variance decreases while the mean moves away from the origin. Our fine-grained pinning lemma shows that the second behavior must happen for a noticeable fraction of conditional distributions.}
\label{fig:pinning-mechanisms}
\end{figure}
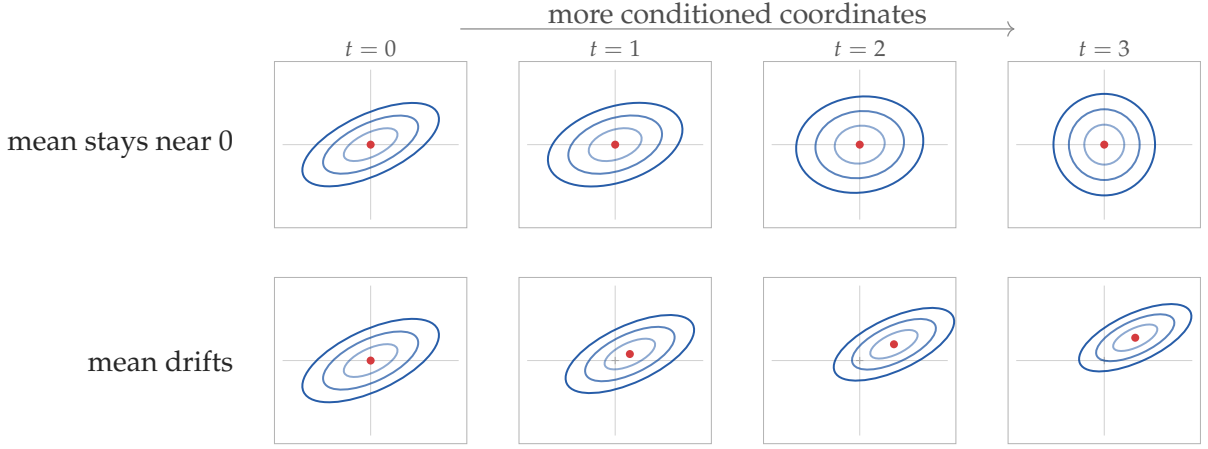

\section{Preliminaries}
\label{sec:prelims}

We now provide an overview of the sum-of-squares proof system, taken from the corresponding section in~\cite{10.1145/3717823.3718218}.
We closely follow the exposition as it appears in the lecture notes of Barak and Steurer~\cite{barak2016proofs}.   

\paragraph{Pseudo-Distributions.}
A discrete probability distribution over $\R^m$ is defined by its probability mass function, $D\from \R^m \to \R$, which must satisfy $\sum_{x \in \mathrm{supp}(D)} D(x) = 1$ and $D \geq 0$.
We extend this definition by relaxing the non-negativity constraint to merely requiring that $D$ passes certain low-degree non-negativity tests.
We call the resulting object a pseudo-distribution.

\begin{definition}[Pseudo-distribution]
A \emph{degree-$\ell$ pseudo-distribution} is a finitely-supported function $D:\R^m \rightarrow \R$ such that $\sum_{x} D(x) = 1$ and $\sum_{x} D(x) p(x)^2 \geq 0$ for every polynomial $p$ of degree at most $\ell/2$, where the summation is over all $x$ in the support of $D$.
\end{definition}
Next, we define the related notion of pseudo-expectation.
\begin{definition}[Pseudo-expectation]
The \emph{pseudo-expectation} of a function $f$ on $\R^m$ with respect to a pseudo-distribution $\mu$, denoted by $\pexpecf{\mu(x)}{f(x)}$,  is defined as
\begin{equation*}
    \pexpecf{\mu(x)}{f(x)} = \sum_{x} \mu(x) f(x).
\end{equation*}
\end{definition}
We use the notation $\pexpecf{\mu(x)}{(1,x_1, x_2,\ldots, x_m)^{\otimes \ell}}$ to denote the degree-$\ell$ moment tensor of the pseudo-distribution $\mu$.
In particular, each entry in the moment tensor corresponds to the pseudo-expectation of a monomial of degree at most $\ell$ in $x$. 

\begin{definition}[Constrained pseudo-distributions]
\label{def:constrained-pseudo-distributions}
Let $\calA = \Set{ p_1\geq 0 , p_2\geq0 , \dots, p_r\geq 0}$ be a system of $r$ polynomial inequality constraints of degree at most $d$ in $m$ variables.
Let $\mu$ be a degree-$\ell$ pseudo-distribution over $\mathbb{R}^m$.
We say that $\mu$ \emph{satisfies} $\calA$ at degree $\ell \ge1$ if for every subset $\calS \subset [r]$ and every sum-of-squares polynomial $q$ such that $\deg(q) + \sum_{i \in \calS } \max\Paren{ \deg(p_i), d} \leq \ell$, $\pexpecf{\mu}{ q \prod_{i \in \calS} p_i } \geq 0$.
Further, we say that $\mu$ \emph{approximately satisfies} the system of constraints $\calA$ if the above inequalities are satisfied up to additive error $\pexpecf{\mu}{ q \prod_{i \in \calS} p_i } \geq -2^{-n^{\ell} } \norm{q} \prod_{i \in \calS} \norm{p_i}$, where $\norm{\cdot}$ denotes the Euclidean norm of the coefficients of the polynomial, represented in the monomial basis.  
\end{definition}

Crucially, there's an efficient separation oracle for moment tensors of constrained pseudo-distributions. 
Below gives the unconstrained statement; the constraint statement follows analogously.

\begin{fact}[\cite{shor1987approach, nesterov2000squared, parrilo2000structured, grigoriev2001complexity}]
    \label{fact:sos-separation-efficient}
    For any $m,\ell \in \N$, the following convex set has a $m^{\bigO{\ell}}$-time weak separation oracle, in the sense of \cite{grotschel1981ellipsoid}\footnote{
        A separation oracle of a convex set $S \subset \R^M$ is an algorithm that can decide whether a vector $v \in \R^M$ is in the set, and if not, provide a hyperplane between $v$ and $S$.
        Roughly, a weak separation oracle is a separation oracle that allows for some $\eta$ slack in this decision.
    }:
    \begin{equation*}
        \Set{  \pexpecf{\mu(x)} { (1,x_1, x_2, \ldots, x_m)^{\otimes \ell } } \Big\vert \text{ $\mu$ is a degree-$\ell$ pseudo-distribution over $\R^m$}}
    \end{equation*}
\end{fact}
This fact, together with the equivalence of weak separation and optimization \cite{grotschel1981ellipsoid} forms the basis of the sum-of-squares algorithm, as it allows us to efficiently approximately optimize over pseudo-distributions. 

Given a system of polynomial constraints, denoted by $ \calA$, we say that it is \emph{explicitly bounded} if it contains a constraint of the form $\{ \|x\|^2 \leq 1\}$. Then, the following fact follows from  \cref{fact:sos-separation-efficient} and \cite{grotschel1981ellipsoid}:

\begin{theorem}[Efficient optimization over pseudo-distributions]
    \label{fact:eff-pseudo-distribution}
There exists an $(m+r)^{O(\ell)} $-time algorithm that, given any explicitly bounded and satisfiable system $ \calA$ of $r$ polynomial constraints in $m$ variables, outputs a degree-$\ell$ pseudo-distribution that satisfies $ \calA$ approximately, in the sense of~\cref{def:constrained-pseudo-distributions}.\footnote{
    Here, we assume that the bit complexity of the constraints in $ \calA$ is $m^{O(1)}$.
}
\end{theorem}

\paragraph{Conditioning and Reweighting.} Finally, we will also require the following \emph{conditioning} operation on pseudo-distributions.
\begin{definition}[Pseudo-distribution reweighting]
    Let $\pE$ be a pseudo-distribution of degree $d$ and let $f$ be a sum-of-squares polynomial of degree $k\leq d$, with $\pE[f] > 0$. Then, $\pE'$, the reweighting of $\pE$ with $f$, is defined by 
    \[ \pE' g = \frac{\pE [fg]}{\pE [f]}\]
    and is a degree $d-k$ pseudo-distribution. Furthermore, $\pE'$ satisfies all constraints that $\pE$ satisfies of degree at most $d-k$.
\end{definition}

Note that if $x$ is an indicator variable and $f = x^2$, then reweighting by $f$ corresponds to conditioning on $x=1$   .

\section{\texorpdfstring{Best Separable State: $1$ vs $1-\eps$}{Best Separable State: 1 vs 1-ε}}
\label{sec:bss-asym}

In this section we prove the $1$ vs $1-\eps$ regime of our main result. In particular, we prove the following theorem:
\begin{theorem}
    \label{thm:asym-bss-close-to-one-main}
    Let $0<\eps<1$ and let $\Pi \in \mathbb{R}^{n^2 \times n^2}$ be a projector to a subspace of $\mathbb{R}^{n^2}$. Then there exists an algorithm that runs in time
    \[
        n^{O\left(\sqrt {\frac n \eps}\right)}
    \]
    and distinguishes between the following cases:
    \begin{itemize}
        \item \textbf{YES:} There exists a nonzero $x, y \in \mathbb{R}^n$ such that $\Pi \left(x \otimes y\right) = x \otimes y$.
        \item \textbf{NO:} For all $x,y \in \mathbb{R}^n$ we have that $\norm{\Pi(x \otimes y)}_2^2 \leq (1-\eps) \cdot \norm{x \otimes y}_2^2$.
    \end{itemize}
    Furthermore, in the \textbf{YES} case, the algorithm produces $\hat{x}, \hat{y}$ such that $\norm{\Pi(\hat{x} \otimes \hat{y})}_2^2 > (1-\eps) \cdot \norm{\hat{x} \otimes \hat{y}}_2^2$.
\end{theorem}

To prove~\cref{thm:asym-bss-close-to-one-main}, we use the constraint set
\begin{align*}
\mathcal A=\left\{
\begin{gathered}
    \norm{x}_2^2=1,\qquad \norm{y}_2^2=1,\\
    \Pi(x\otimes y)=x\otimes y
\end{gathered}
\right\}.
\end{align*}
Since we will work over a discretized domain it will be potentially impossible to satisfy the constraints above, even in the YES case. Thus, we instead replace $\mathcal A$ with a version which allows $1/\poly(n)$ slack:
\begin{align}
\mathcal A'(\beta)=\left\{
\begin{gathered}
    1-\frac1{n^\beta}\leq\norm{x}_2^2\leq1+\frac1{n^\beta},\\
    1-\frac1{n^\beta}\leq\norm{y}_2^2\leq1+\frac1{n^\beta},\\
    \norm{(I-\Pi)(x\otimes y)}_2^2\leq\frac1{n^{2\beta}}
\end{gathered}
\right\}. \label{eq:eps_constraints}
\end{align}

\begin{mdframed}
  \begin{algorithm}[Best Separable State: $1$ vs $1-\eps$]
    \label{algo:asym-final-bss}\mbox{}
    \begin{description}
    \item[Input:] A projector $\Pi\in\mathbb R^{n^2\times n^2}$ and
    $0<\eps<1$.
    \item[Operations:]\mbox{}
    \begin{enumerate}
        \item If $\eps<1/n$, search over all pairs $u,v$ in an
        $(\eps/20)$-net of $\mathbb S^{n-1}$ and output YES iff
        $\norm{\Pi(u\otimes v)}_2^2>1-\eps$ for some pair.
        \item Otherwise, let $\gamma=\sqrt\eps/4$ and choose $\beta=4$.
        \item Let $\Sigma$ be a grid of $[-1,1]$ of mesh size a sufficiently
        small constant multiple of $n^{-1/2-\beta}$.
        \item Let $T=\lceil C\sqrt n/\gamma\rceil$ for a sufficiently large
        absolute constant $C$. Check feasibility of the degree-$(2T+8)$
        sum-of-squares relaxation over $\Sigma$ satisfying
        $\mathcal A'(\beta)$.
        \item If the program is feasible, output YES. Otherwise, output NO.
    \end{enumerate}
    \end{description}
  \end{algorithm}
\end{mdframed}

In order to show that~\cref{algo:asym-final-bss} is correct, we show that whenever the program is feasible we can round the relaxation to $\hat{x}, \hat {y}$ such that $\norm{\Pi(\hat{x} \otimes \hat{y})}_2 > (1-\eps) \norm{\hat{x} \otimes \hat{y}}_2$. This means that in the NO case, the program cannot be feasible. The rounding algorithm used is a variant of global correlation rounding. We record it below as simply brute forcing over all possible choices of conditionings and indices to condition on; this is because the precise choice of coordinates to condition on and how to sample their values is not straightforward to describe. We defer further discussion of how to show the existence of a good set of coordinates and values to the proof of~\cref{lem:asymmetric-pinning}.

\begin{mdframed}
  \begin{algorithm}[Best Separable State Rounding]
    \label{algo:asym-final-bss-rounding}\mbox{}
    \begin{description}
    \item[Input:] A degree-$O(T)$ pseudodistribution satisfying
    $\mathcal A'(\beta)$.
    \item[Operations:]\mbox{}
        For each set $S_{x} \subseteq [n]$ and $S_{y} \subseteq [n]$ of indices of size at most $T$ and pinnings $\wh{x}_{S_x}$ and $\wh{y}_{S_y}$:
        \begin{enumerate}
            \item Let $\pE$ be the pseudodistribution obtained by conditioning $\pE$ on the event that $x_{S_x} = \wh{x}_{S_x}$ and $y_{S_y} = \wh{y}_{S_y}$.
            \item Let $m=\pE_r x$ and $q=\pE_r y$ at the resulting branch.
        \end{enumerate}
    \item[Output:] $m/\norm m_2$ and $q/\norm q_2$ with the highest objective value among all iterations of the loop.
    \end{description}
  \end{algorithm}
\end{mdframed}

\subsection{The rounding success condition}
We now describe a sufficient condition for the final mean rounding step to succeed in~\cref{algo:asym-final-bss-rounding}.

\begin{lemma}[Mean rounding]
\label{lem:asym-final-mean-rounding}
Let $\pE$ satisfy $\mathcal A'(\beta)$, and suppose that
\[
    \norm{\Cov(x,y)}_F
    \leq\gamma\norm{\pE x}_2\norm{\pE y}_2.
\]
Write $m=\pE x$ and $q=\pE y$. Then
\[
    \norm{\Pi(m\otimes q)}_2^2
    \geq
    \left(1-\left(\gamma+
    \frac{n^{-\beta}}{\norm m_2\norm q_2}\right)^2\right)
    \norm{m\otimes q}_2^2.
\]
\end{lemma}

\begin{proof}
Let $M=\pE(x\otimes y)$. By the definition of cross-covariance,
\[
    M=m\otimes q+\Cov(x,y).
\]
Moreover, Jensen's inequality and the residual constraint give
\[
    \norm{(I-\Pi)M}_2
    \leq
    \left(\pE\norm{(I-\Pi)(x\otimes y)}_2^2\right)^{1/2}
    \leq n^{-\beta}.
\]
Therefore
\begin{align*}
    \norm{(I-\Pi)(m\otimes q)}_2
    &\leq\norm{(I-\Pi)M}_2+\norm{\Cov(x,y)}_F\\
    &\leq n^{-\beta}+\gamma\norm m_2\norm q_2.
\end{align*}
Since $\Pi$ is an orthogonal projector,
\[
    \norm{\Pi(m\otimes q)}_2^2
    =\norm{m\otimes q}_2^2-\norm{(I-\Pi)(m\otimes q)}_2^2.
\]
Substituting the previous bound proves the lemma.
\end{proof}

\subsection{A new pinning lemma}
\label{sec:eps-pinning-lemma}

	We now establish that we can achieve the desired ``win condition'' via $O(\sqrt{n/\eps})$ rounds of conditioning.

	\begin{lemma}
		\label{lem:asymmetric-pinning}
		Let $0 < \gamma \le 1$, and let $\pE$ be a pseudodistribution of degree $\Omega\left( \frac{\sqrt{n}}{\gamma} \right)$ on random variables $(x,y)$, satisfying the constraints in \eqref{eq:eps_constraints}. Then, there is a pinning of $O\left( \frac{\sqrt{n}}{\gamma} \right)$ coordinates such that denoting by $\pE'$ the conditional pseudodistribution obtained from these pinnings
		\begin{enumerate}[label=(\alph*)]
			\item $\displaystyle \left\| \wt{\Cov}'(x,y) \right\|_{F} \le \gamma \cdot \|\pE' x\| \cdot \|\pE' y\|$.
			\item $\|\pE' x\|$ and $\|\pE' y\|$ are bounded from below by $\frac{\gamma}{\sqrt{n}}$.
		\end{enumerate}
	\end{lemma}

	\begin{remark}
		The more important of the two conditions is the first one. The second condition is needed only to control the error from our $1/\poly(n)$ slack constraints. It turns out that it is also naturally established in the process of bound $\norm{\widetilde{\Cov}'(x,y)}_F$.
	\end{remark}

		\Cref{lem:asymmetric-pinning} will follow from two lemmas. The first is easy, and establishes that it is possible to pin a small number of coordinates in order to ensure the required lower bounds on the norms of the mean.

		\begin{lemma}
			\label{lem:mean-anchoring}
			Let $0 < Q \ll 1$. There is a set $I_{x} \subseteq [n]$ and a pinning $a_{j}$ of the coordinates $x_{j}$ for $j \in I_{x}$ such that $\sum a_{j}^2 \ge Q$, and $|I_{x}| = O\left(nQ\right)$. The same claim holds for $y$.
		\end{lemma}
		\begin{proof}
			Suppose $Q < \frac{1}{4}$, say. We grow the set $I_{x}$ one coordinate at a time. While $\sum a_{j}^2 \le Q$, we have $\sum_{i \not\in I_{x}} \pE x_{i}^2 \ge \frac{1}{2}$, so there exists a choice of coordinate $i \not\in I_{x}$ with $\pE x_{i}^2 \ge \frac{1}{2n}$, and thus a pinning $a_{i}$ occuring with nonzero probability such that $a_i^2 \ge \frac{1}{2n}$. Repeating this at most $2Qn$ times completes the proof.
		\end{proof}

		The second establishes the more interesting claim. Define the potential function $\Phi : \R^n \times \R^n \to \R$ by
		\[ \Phi(u,v) = 2 \left(\|u\|^2 + \|v\|^2 + \|u\| \cdot \|v\|\right) \mcom \]
		and denote by $C$ the cross-covariance
		\[ C = \pE \left( x - \pE x \right) \left( y - \pE y \right)^\top \mper \]

		\begin{restatable}{lemma}{potentialgrowth}
			\label{lem:potential-growth-43}
			Let $\pE$ be a pseudodistribution on variables $(x,y)$ satisfying the constraints \eqref{eq:eps_constraints}. Further suppose that all coordinates in some sets of indices $I_{x}, I_{y} \subseteq [n]$ have already been pinned such that they contribute squared $\ell_2$ mass of at least $Q$. Let $m_x$ and $m_y$ be the means of this pseudodistribution. Also suppose that
			\[ \| C \|_{F} > \gamma \cdot \|m_x\| \cdot \|m_y\| \mper \]
			Then, there exists a pinning such that if $\wt{m}_x$ and $\wt{m}_y$ are the updated means after this pinning, then
			\[ \Phi\left( \wt{m}_x , \wt{m}_y \right) - \Phi\left( m_x , m_y \right) \gtrsim \gamma^2 Q^{2/3} \cdot \Phi\left( m_x , m_y \right)^{4/3} \mper \]
		\end{restatable}

		We defer a discussion on the choice of potential function and the above lemma to \Cref{subsec:potential-growth}.

		\begin{proof}[Proof of \Cref{lem:asymmetric-pinning}]
			Begin by pinning coordinates of both $x$ and $y$ according to \Cref{lem:mean-anchoring}, for some $Q$ to be decided. Let $m_x^0$ and $m_y^0$ be the means at this step. Following this, while the desideratum $\|C_{F}\| \le \gamma \|m_x\| \|m_y\|$ is not satisfied, progressively pin coordinates according to \Cref{lem:potential-growth-43}, and let $m_x^t$ and $m_y^t$ be the means after $t$ steps of this pinning. We shall show that this can happen for at most $O\left( \frac{1}{Q\gamma^2} \right)$ steps.
			Set $\Phi_{t} = \Phi\left( m_x^t , m_y^t \right)$. By the potential growth lemma, we have for some constant $c$ that
			\[ \Phi_{t+1} \ge \Phi_{t} + c\gamma^2 Q^{2/3} \cdot \Phi_{t}^{4/3} \mcom \]
			so
			\begin{align*}
				\Phi_{t+1}^{-1/3} &\le \Phi_{t}^{-1/3} \cdot \left( 1 + c\gamma^2 Q^{2/3} \cdot \Phi_t^{1/3} \right)^{-1/3} \\
					&\le \Phi_t^{-1/3} - c\gamma^2Q^{2/3} \\
				\Phi_t^{-1/3} - \Phi_{t+1}^{-1/3} &\ge c\gamma^2Q^{2/3} \mper
			\end{align*}
			Note now that $\Phi_0 \gtrsim Q$, by the lower bound on the norm of $m_x^0$. That is, $\Phi_0^{-1/3} \lesssim Q^{-1/3}$. Since $\Phi$ is non-negative, this implies that this can proceed for at most $O\left( \frac{ Q^{-1/3} }{ \gamma^2 Q^{2/3} } \right) = O\left( \frac{1}{\gamma^2 Q} \right)$ many steps.

			Now, we performed $O(Qn)$ rounds of conditioning in the first phases, and $O\left( \frac{1}{\gamma^2 Q} \right)$ many steps of conditioning in the second. Balancing the two by choosing $Q \asymp \frac{1}{\gamma \sqrt{n}}$ completes the proof.
		\end{proof}

	\subsection{A potential growth lemma}
		\label{subsec:potential-growth}

		The primary result of this section is the following.

		\potentialgrowth*

		The structure of $\Phi$ has two components that behave in different ways: the $\|u\|^2 + \|v\|^2$ part results in growth when some coordinate has large variance, and as we will see shortly, the mixed term rewards movement \emph{orthogonal} to the existing mean. In fact, the only ``orthogonal movement'' we will use will be the part that is orthogonal to coordinates $I_{x} \cup I_{y}$ that have already been pinned.

		In particular, the first part, $\|u\|^2 + \|v\|^2$, suffices to establish a potential growth of order $\gamma^2 Q \Phi$---this will give a bound on the degree of the required SDP of $O\left( \sqrt{ \frac{n}{\eps} \cdot \log n } \right)$. In order to remove the extraneous $\log$ factors here and improve this to $O\left( \sqrt{ \frac{n}{\eps} } \right)$, we will require the better potential growth above.
		
		Let us start by noting some properties of $\Phi$. The first is a consequence of its convexity.

		\begin{lemma}
			\label{lem:noise-helps}
			Let $\bx$ and $\by$ be (possibly correlated) mean-$0$ random variables on $\R^n$. Then,
			\[ \E \Phi\left( u + \bx , v + \by \right) - \Phi\left( u , v \right) \ge \E \left[\|\bx\|^2 + \|\by\|^2\right] \mper \]
		\end{lemma}
		\begin{proof}
			Note that $(u,v) \mapsto (\|u\|+\|v\|)^2$ is convex. We thus have
			\begin{align*}
				\E \Phi\left( u + \bx , v + \by \right) &= \E \left[\left( \|u + \bx\| + \|v + \by\| \right)^2 + \|u + \bx\|^2 + \|v + \by\|^2\right] \\
					&\ge \left( \|u\| + \|v\| \right)^2 + \E \left[\|u + \bx\|^2 + \|v + \by\|^2\right] \\
					&= \Phi(u,v) + \E\left[ \|\bx\|^2 + \|\by\|^2 \right] \mper \qedhere
			\end{align*}
		\end{proof}

		The second establishes that if one of the means has grown large, and the other moves upon conditioning in some direction orthogonal to its (before-conditioning) value, then one may obtain stronger growth bounds.

		\begin{lemma}
			\label{lem:mismatched-noise-helps-more}
			Let $\bx$ and $\by$ be (possibly correlated) mean-$0$ random variables on $\R^n$, and suppose that almost surely, $\left\| \Pi_{u^\perp} \bx\right\|_2^2 \ge \theta \cdot \|\bx\|^2$. Then,
			\[ \E \Phi\left( u + \bx , v + \by \right) - \Phi\left( u , v \right) \ge \theta \|v\| \cdot \E \left[ \frac{\|\bx\|^2}{\|u + \bx\|} \right] \mper \] 
			In particular, if there is some collection $I$ of indices such that (i) we almost surely have $\bx_{I} = \by_{I} = 0$, and (ii) $\|u_{I}\|^2 \ge Q$, then we may take $\theta = \frac{Q}{\|u\|^2}$.
		\end{lemma}
		\begin{proof}
			We have
			\begin{align*}
				\E \Phi\left( u + \bx , v + \by \right) - \Phi\left( u , v \right) &\ge \E \left[\left( \|u+\bx\| + \|v+\by\| \right)^2\right] - \left( \|u\| + \|v\| \right)^2 \\
					&\ge 2 (\|u\|+\|v\|) \cdot \E \left( \|u+\bx\| - \|u\| + \|v+\by\| - \|v\| \right) \\
					&\ge 2 \|v\| \cdot \E \left( \|u+\bx\| - \|u\| \right) \mcom
			\end{align*}
			where the final inequality uses the convexity of $v \mapsto \|v\|$. Continuing,
			\begin{align*}
				\E \left( \|u+\bx\| - \|u\| \right) &= \E \left( \|u+\bx\| - \left\langle \frac{u}{\|u\|} , u+\bx \right\rangle \right) \\
					&= \E \frac{ \|u+\bx\|^2 - \left\langle \frac{u}{\|u\|} , u+\bx \right\rangle^2 }{ \|u+\bx\| + \left\langle \frac{u}{\|u\|} , u+\bx \right\rangle } \\
					&\ge \frac{1}{2} \cdot \E \frac{ \|u+\bx\|^2 - \left\langle \frac{u}{\|u\|} , u+\bx \right\rangle^2 }{ \|u+\bx\| } \\
					&= \frac{1}{2} \cdot \E \frac{ \left\| \Pi_{u^\perp} \bx \right\|^2 }{ \|u+\bx\| } \\
					&\ge \frac{\theta}{2} \cdot \E \frac{\|\bx\|^2}{\|u+\bx\|} \mcom
			\end{align*}
			as desired.
			For the second part of the lemma, the assumption on $I$ implies that
			\[ \left\| \Pi_{u^\perp} \bx \right\|^2 = \|\bx\|^2 - \left\langle \frac{u}{\|u\|} , \bx \right\rangle^2 \ge \left( 1 - \frac{\|u_{I^c}\|^2}{\|u\|^2} \right) \cdot \|\bx\|^2 \ge \frac{Q}{\|u\|^2} \cdot \|\bx\|^2 \mcom \]
			completing the proof.
		\end{proof}

		For the remainder of this section, let us assume without loss of generality that $\|m_x\| \le \|m_y\|$. Our pinning will pin a coordinate of $y$. For a given pinning $y_j \gets a$, denote $m_x^{j \gets a}$ (resp. $m_y^{j \gets a}$) as the new $x$-mean (resp. $y$-mean) after this pinning. Let us also denote
		\[ \Delta = \max_{j,a}\left\{ \Phi\left( m_x^{j \gets a} , m_y^{j \gets a} \right) - \Phi\left( m_x , m_y \right) \right\} \mper \]

		Let us also explicitly write the elements of the cross-covariance $C$. Its $j$th column is given by
		\[ C_{j} = \pE_{a \sim \mu_{y_j}} \left[ a \left( m_{x}^{j \gets a} - m_{x} \right) \right] \mper \]
		The following lemma controls various quantities related to these vectors.

		\begin{lemma}
			\label{lem:tilted-control}
			We have
			\begin{equation}
				\label{eq:centeredness-tilted-pinning}
				\left\| \sum_{j} \pE_{a \sim \mu_{y_j}} a^2 \left( m_x^{j \gets a} - m_x \right) \right\| = O(n^{-\beta})
			\end{equation}
			and
			\begin{equation}
				\label{eq:secondmoment-tilted-pinning}
				\sum_{j} \pE_{a \sim \mu_{y_j}} a^2 \left\| \left( m_x^{j \gets a} - m_x \right) \right\|^2 \lesssim \Delta + O(n^{-\beta}) \mper
			\end{equation}
		\end{lemma}

		The above will be a consequence of the following (slightly mysterious) measure decomposition when the distributions of $x$ and $y$ are supported exactly on the sphere $\mathbb{S}^{n-1}$. Pick a coordinate $j \in [n]$ with probability $\E y_j^2$, and then pin it to $a$ with probability $\frac{\Pr[y_j = a] \cdot a^2}{\E y_j^2}$. Observe that this does \emph{not} lie in the standard canon of pinning algorithms of the form ``pick a coordinate according to some distribution, then pin the value of that coordinate according to the associated marginal.''

		Let us verify that this is indeed a measure decomposition, before providing intuition. We have
		\begin{align}
			\sum_{j,a} \E\left[ f(x,y) \mid y_j = a \right] \cdot \frac{\Pr[y_j = a] \cdot a^2}{\E y_j^2} \cdot \E[y_j^2] &= \sum_{j,a} \E\left[ f(x,y) y_j^2 \mid y_j = a \right] \cdot \Pr[y_j = a] \nonumber \\
				&= \sum_{j} \E\left[ f(x,y) y_j^2 \right] \nonumber \\
				&= \E\left[ f(x,y) \right] \label{eq:magic-l2-measure-decomposition}
		\end{align}
		The above is more clearly seen to be a measure decomposition when interpreted in the following way. Sample $y$, pick a coordinate $j \sim p_{y}$ (in this case, $p_{y}(j) = y_j^2$, but this remains a measure decomposition even otherwise), and pin $\by_{j} = y_{j}$. Indeed, this can be viewed as the ``first step'' of the measure decomposition which eventually pins all the coordinates of $y$. When $p_{y}(j) = y_{j}^2$, both the marginal probability of $j$ (that is, $\pE y_j^2$) \emph{and} the marginal on $a$ conditioned on the coordinate $j$ being picked (that is, with probability $\frac{\Pr[y_j = a] \cdot a^2}{\E y_j^2}$) have particularly simple interpretations. We now move to the proof of the above lemma.

		\begin{proof}
			By essentially the manipulations in \eqref{eq:magic-l2-measure-decomposition}, picking $f(x,y) = x$,
			\[ \frac{\sum_{j} \pE_{a \sim \mu_{y_{j}}}\left[a^2 m_{x}^{j \gets a}\right]}{\pE \|y\|^2} = \frac{\pE\left[ x \|y\|^2 \right]}{\pE \|y\|^2} \mper \]
			\eqref{eq:eps_constraints} establishes that $\pE \vDash \|y\|^2 = 1 + O(n^{-\beta})$, which yields \eqref{eq:centeredness-tilted-pinning}.
			Towards the second inequality,
			\begin{align*}
				\Delta &\ge \frac{\sum_{j} \pE_{a \sim \mu_{y_j}} a^2 \Phi\left( m_x^{j \gets a} , m_y^{j \gets a} \right)}{\pE \|y\|^2} - \Phi(m_x, m_y) \mper
			\end{align*}
			Denote $m_x' = \sum_{j} \pE_{a \sim \mu_{y_j}} a^2 m_x^{j \gets a}$, and define $m_y'$ similarly, so $\|m_x' - m_x\| = O(n^{-\beta})$. Consequently, $\left| \Phi(m_x',m_y') - \Phi(m_x,m_y) \right| = O(n^{-\beta})$. As a result, by \Cref{lem:noise-helps},
			\begin{align*}
				\Delta &\ge \frac{\sum_{j} \pE_{a \sim \mu_{y_j}} a^2 \Phi\left( m_x^{j \gets a} , m_y^{j \gets a} \right)}{\pE \|y\|^2} - \Phi(m_x',m_y') - O(n^{-\beta}) \\
					&\ge \frac{\sum_{j} \pE_{a \sim \mu_{y_j}} \left\| m_x^{j \gets a} - m_x \right\|^2}{\pE \|y\|^2} - O(n^{-\beta}) \\
					&= \sum_{j} \pE_{a \sim \mu_{y_j}} \left\| m_x^{j \gets a} - m_x \right\|^2 - O(n^{-\beta}) \mper \qedhere
			\end{align*}
		\end{proof}

		 We will also use the same lower bound on $\Delta$ that was employed in the weaker $O(\sqrt{n \log n / \epsilon})$ argument:

		\begin{lemma}
			\label{lem:worse-linear-cov-bound}
			We have
			\[ \|C\|_F^2 \le \Delta + O(n^{-\beta}) \mper \]
		\end{lemma}
        (For the proof, see the warmup in the technical overview, in particular \eqref{eq:tech-overview-cov-progress}.)

		As mentioned, the above bound already yields non-trivial control on the growth of the potential---ignoring the $n^{-\beta}$ terms, it gives $\Delta \gtrsim \gamma^2 \|m_x\|^2 \|m_y\|^2 \gtrsim \gamma^2 Q \Phi(m_x,m_y)$. In our final proof, we will only use this to obtain an upper bound on how large the smaller mean $\|m_x\|$ can be.

		With the above in hand, we may prove the main input to \Cref{lem:potential-growth-43}.

		\begin{lemma}
			\label{lem:better-nonlinear-cov-bound}
			We have
			\[ \|C\|_F^2 \le \Delta \cdot \frac{\|m_x\|^2}{Q\|m_y\|} \cdot \sqrt{ \Delta + \|m_x\|^2 } \]
		\end{lemma}
		Among the terms within the square-root, the $\|m_x\|^2$ will be the relevant one.
		\begin{proof}
			First off, for any coordinate $j$,
			\begin{equation}
				\label{eq:mismatched-eq1}
				\pE_{a \sim \mu_{y_j}} \left[ \frac{\|m_x^{j \gets a} - m_x\|^2}{\|m_x^{j \gets a}\|} \right] \le \Delta \cdot \frac{\|m_x\|^2}{Q \|m_y\|} \mper
			\end{equation}
			Indeed, this follows immediately from \Cref{lem:mismatched-noise-helps-more} with the observation that $m_x^{j \gets a} - m_x$ is a $0$-mean random variable.

			Let $C_j$ be the $j$th column of $C$. As observed,
			\[ C_{j} = \pE_{a \sim \mu_{y_j}} a \left( m_x^{j \gets a} - m_x \right) \mper \]
			By the weighted Cauchy--Schwarz inequality,
			\begin{align*}
				\|C_j\|^2 &\le \left(\pE_{a \sim \mu_{y_j}} \frac{\left\| m_x^{j \gets a} - m_x \right\|^2}{\|m_x^{j \gets a}\|}\right) \cdot \left( \pE_{a \sim \mu_{y_j}} a^2 \cdot \|m_x^{j \gets a}\| \right) \\
					&\le \Delta \cdot \frac{\|m_x\|^2}{Q\|m_y\|} \cdot \left( \pE_{a \sim \mu_{y_j}} a^2 \cdot \|m_x^{j \gets a}\| \right)
			\end{align*}
			Therefore, summing over $j$,
			\begin{align*}
				\|C\|_F^2 &\le \Delta \cdot \frac{\|m_x\|^2}{Q\|m_y\|} \cdot \sum_{j} \pE_{a \sim \mu_{y_j}} a^2 \cdot \|m_x^{j \gets a}\| \\
					&\le \Delta \cdot \frac{\|m_x\|^2}{Q\|m_y\|} \cdot \left( \pE \|y\|^2 \cdot \sum_j \pE_{a \sim \mu_{y_j}} a^2 \cdot \|m_x^{j \gets a}\|^2\right)^{1/2}
			\end{align*}
			The constraints in \eqref{eq:eps_constraints} yield that $\pE \|y\|^2 = 1 + O(n^{-\beta})$. The numerator of the term within the square-root may be rewritten as
			\begin{align*}
				\sum_j \pE_{a \sim \mu_{y_j}} a^2 \cdot \|m_x^{j \gets a}\|^2 &= \sum_j \pE_{a \sim \mu_{y_j}} a^2 \cdot \left\|m_x + \left( m_x^{j \gets a} - m_x \right) \right\|^2 \\
					&= \left(\pE \|y\|^2\right) \cdot \|m_x\|^2 + 2 \left\langle m_x , \sum_{j} \pE_{a \sim \mu_{y_j}} a^2 m_x^{j \gets a} - m_x \right\rangle \\
					&\qquad\qquad+ \sum_{j} \pE_{a \sim \mu_{y_j}} a^2 \left\| m_x^{j \gets a} - m_x \right\|^2
			\end{align*}
			\eqref{eq:centeredness-tilted-pinning} in \Cref{lem:tilted-control} bounds the second quantity by $O(n^{-\beta})$, and \eqref{eq:secondmoment-tilted-pinning} bounds the third by $O(\Delta + n^{-\beta})$. Putting the pieces together, we thus have
			\[ \|C\|_F^2 \lesssim \Delta \cdot \frac{\|m_x\|^2}{Q\|m_y\|} \cdot \sqrt{ \Delta + \|m_x\|^2} \mper \qedhere \]
		\end{proof}

		Finally, we may prove \Cref{lem:potential-growth-43}.

		\begin{proof}[Proof of \Cref{lem:potential-growth-43}]
			\Cref{lem:better-nonlinear-cov-bound} coupled with our lower bound assumption on $\|C\|_F$ imply that
			\[ \Delta \cdot \frac{\|m_x\|^2}{Q\|m_y\|} \cdot \sqrt{ \Delta + \|m_x\|^2 } \gtrsim \gamma^2 \cdot \|m_x\|^2 \|m_y\|^2 \mper \]
			Rearranging,
			\[ \Delta \cdot \sqrt{ \Delta + \|m_x\|^2 } \gtrsim \gamma^2 \cdot Q \cdot \|m_y\|^3 \mper \]
			On the other hand, \Cref{lem:worse-linear-cov-bound} yields that
			\[ \Delta \gtrsim \gamma^2 \cdot \|m_x\|^2 \|m_y\|^2 \mper \]
			Indeed, the error of $O(n^{-\beta})$ is negligible in comparison to the Frobenius norm of the covariance, since the failure of the win condition and the mean anchoring imply that $\|C\|_{F}^2 \gtrsim n^{-O(1)}$.
			As a result,
			\[ \|m_x\|^2 \lesssim \frac{\Delta}{\gamma^2 \|m_y\|^2} \mper \]
			Substituting this into the previous string of inequalities,
			\[ \Delta \cdot \sqrt{ \frac{\Delta}{\gamma^2 \|m_y\|^2} } \gtrsim \gamma^2 \cdot Q \cdot \|m_y\|^3 \mcom \]
			and rearranging,
			\[ \Delta^{3/2} \ge \gamma^3 Q \cdot \|m_y\|^4 \gtrsim \gamma^3 Q \cdot \Phi(m_x,m_y)^{2} \mper \]
			Raising both sides to the $(2/3)$rd power completes the proof.
		\end{proof}

\subsection{Proof of the main theorem}

We now combine the pinning lemma from the previous section and the rounding success condition to prove the main theorem.

\begin{proof}[Proof of~\cref{thm:asym-bss-close-to-one-main}]
We first consider the case where $\eps > 1/n$ (we will return to the case where $\eps < 1/n$ later). First, we show that in the YES case the SDP relaxation must be feasible. Normalize a YES witness so that
$\norm{x}_2=\norm{y}_2=1$, and choose $\hat x,\hat y\in\Sigma^n$ with
\[
    \norm{x-\hat x}_2,\norm{y-\hat y}_2\leq c n^{-\beta}
\]
for a sufficiently small absolute constant $c$.  Then
\[
    \left|\norm{\hat x}_2^2-1\right|,
    \left|\norm{\hat y}_2^2-1\right|
    \leq n^{-\beta}.
\]
Moreover,
\[
    \norm{\hat x\otimes\hat y-x\otimes y}_2
    \leq
    \norm{(\hat x-x)\otimes\hat y}_2
    +\norm{x\otimes(\hat y-y)}_2
    =O(n^{-\beta}).
\]
Since $(I-\Pi)(x\otimes y)=0$, decreasing $c$ if necessary gives
\[
    \norm{(I-\Pi)(\hat x\otimes\hat y)}_2\leq n^{-\beta}.
\]
Thus the point mass at $(\hat x,\hat y)$ satisfies
$\mathcal A'(\beta)$, and the relaxation is feasible in the
YES case.

It remains to prove soundness.  Suppose that the relaxation is feasible.
Apply~\cref{lem:asymmetric-pinning} to the pseudodistribution induced by
the feasible solution. Thus, one of the branches
enumerated by~\cref{algo:asym-final-bss-rounding} satisfies
\[
    \norm{\Cov_r(x,y)}_F
    \leq
    \gamma\norm{\pE_r x}_2\norm{\pE_r y}_2 \quad \text{and} \quad \norm{\pE_r x}_2,\norm{\pE_r y}_2
    \geq\frac{\gamma}{\sqrt n}.
\]
Set $m=\pE_r x$ and $q=\pE_r y$.  The second conclusion of the pinning
lemma gives
\[
    \frac{n^{-\beta}}{\norm m_2\norm q_2}
    \leq
    \frac{n^{1-\beta}}{\gamma^2}.
\]
Recall that $\beta=4$, $\gamma=\sqrt\eps/4$, and $\eps\geq1/n$.  Therefore
\[
    \frac{n^{1-\beta}}{\gamma^2}
    =
    \frac{16}{n^3\eps}
    \leq\gamma
\]
for all sufficiently large $n$. 
Applying~\cref{lem:asym-final-mean-rounding}, we conclude that
\begin{align*}
    \norm{\Pi(m\otimes q)}_2^2
    \geq
    \left(1-(2\gamma)^2\right)\norm{m\otimes q}_2^2=
    \left(1-\frac{\eps}{4}\right)\norm{m\otimes q}_2^2>
    (1-\eps)\norm{m\otimes q}_2^2.
\end{align*}
Thus the normalized mean pair violates the NO condition.  Since the rounding algorithm tests every
conditioning branch of the required length, it finds a pair with at least
this objective value.  Consequently, the relaxation cannot be feasible in
the NO case.

In order to conclude the case where $\eps > 1/n$, it remains to address the runtime. The grid has size $\abs{\Sigma}=n^{O(1)}$, and
\[
    T=O\left(\frac{\sqrt n}{\gamma}\right)
    =O\left(\sqrt{\frac n\eps}\right).
\]
The degree-$O(T)$ sum-of-squares feasibility problem can therefore be solved
in time
\[
    n^{O(T)}
    =
    n^{O\left(\sqrt{n/\eps}\right)}.
\]
Enumerating the conditioning branches in
\cref{algo:asym-final-bss-rounding} has the same running time. 

	We now return to the case where $\eps < 1/n$.
	In this case, the algorithm searches over all pairs $u,v$ in an
	$(\eps/20)$-net of $\mathbb S^{n-1}$. We first show completeness. Normalize
	a YES witness so that $\norm{x}_2=\norm{y}_2=1$, and choose $u,v$ in the
	net such that
	\[
	    \norm{x-u}_2,\norm{y-v}_2\leq\frac{\eps}{20}.
	\]
	Then
	\begin{align*}
	    \norm{u\otimes v-x\otimes y}_2
	    &\leq
	    \norm{(u-x)\otimes v}_2+
	    \norm{x\otimes(v-y)}_2\\
	    &\leq\frac{\eps}{10}.
	\end{align*}
	Since $\Pi(x\otimes y)=x\otimes y$ and $\Pi$ is a contraction, it follows
	that
	\[
	    \norm{\Pi(u\otimes v)}_2
	    \geq
	    1-\frac{\eps}{10}.
	\]
	Thus
	\[
	    \norm{\Pi(u\otimes v)}_2^2
	    \geq
	    \left(1-\frac{\eps}{10}\right)^2
	    >
	    1-\eps,
	\]
	so the algorithm outputs YES. On the other hand, in the NO case every
	pair of unit vectors $u,v$ satisfies
	\[
	    \norm{\Pi(u\otimes v)}_2^2\leq1-\eps,
	\]
	and therefore the algorithm outputs NO.

	Finally, an $(\eps/20)$-net of $\mathbb S^{n-1}$ has size
	$(O(1/\eps))^n$. Hence, the search over all pairs in the net takes time
	\[
	    (O(1/\eps))^{2n}.
	\]
	Since $\eps<1/n$, we have
	\[
	    (O(1/\eps))^{2n}
	    =
	    n^{O\left(\sqrt{n/\eps}\right)}.
	\]
	This proves the theorem in the remaining regime and completes the proof.
\end{proof}

\section{Best Separable State: $1$ vs $q/n$}

\begin{restatable}{theorem}{lowsoundnessthm}
    \label{thm:asym-bss-low-soundness}
    Let $1 < q < n/2$ and let $\Pi \in \mathbb{R}^{n^2 \times n^2}$ be a projector to a subspace of $\mathbb{R}^{n^2}$. Then there exists an algorithm that runs in time
    \[
        n^{O\left(\sqrt {q}\right)}
    \]
    and distinguishes between the following cases:
    \begin{itemize}
        \item \textbf{YES:} There exists a nonzero $x, y \in \mathbb{R}^n$ such that $\Pi \left(x \otimes y\right) = x \otimes y$.
        \item \textbf{NO:} For all $x,y \in \mathbb{R}^n$ we have that $\norm{\Pi(x \otimes y)}_2^2 \leq \frac q n \cdot \norm{x \otimes y}_2^2$.
    \end{itemize}
    Furthermore, in the \textbf{YES} case, the algorithm produces $\hat{x}, \hat{y}$ such that $\norm{\Pi(\hat{x} \otimes \hat{y})}_2^2 > \frac q n \cdot \norm{\hat{x} \otimes \hat{y}}_2^2$.
\end{restatable}

As in the previous section, we consider the natural SoS relaxation of BSS.
Since the rounding algorithm conditions on coordinate values, we work over a
finite grid. Since over this grid it may not be possible to satisfy the canonical constraint set exactly, even in the YES case, we instead use the relaxed constraint set
\begin{align}
\mathcal A'(\beta)=\left\{
\begin{gathered}
    1-n^{-\beta}\leq\norm{x}_2^2\leq1+n^{-\beta},\qquad
    1-n^{-\beta}\leq\norm{y}_2^2\leq1+n^{-\beta},\\
    \norm{(I-\Pi)(x\otimes y)}_2^2\leq n^{-2\beta}
\end{gathered}
\right\}. \label{eq:constraints}
\end{align}
In our application below we take $\beta=2$, and let $\Sigma$ be a grid of $[-1,1]$
whose mesh is a sufficiently small constant multiple of
$n^{-1/2-\beta}$. 

The decision algorithm checks feasibility of the degree-$C\sqrt q$ SoS
relaxation satisfying~\eqref{eq:constraints}, where $C$ is a sufficiently
large absolute constant.  This is the same feasibility test as in the
$1$ versus $1-\eps$ regime; only the degree and the rounding argument differ.

We once again show that the relaxation must be infeasible in the NO case by providing a rounding algorithm. However, compared to the $1$ vs $1-\eps$ regime, the final rounding step is slightly more involved. We use a combination of coordinate pinnings and reweights, but use multiple different final rounding steps. There are two key differences from the algorithm in~\cref{sec:bss-asym}: we require reweights in addition to coordinate pinnings, and we also consider multiple final rounding steps.

\begin{mdframed}
  \begin{algorithm}[Low-soundness BSS rounding]
    \label{algo:low-soundness-rounding}\mbox{}
    \begin{description}
    \item[Input:] A degree-$C\sqrt q$ pseudodistribution satisfying
    $\mathcal A'(\beta)$.
    \item[Operations:] For every branch obtained by performing at most
    $C\sqrt q$ operations, each of which is either a pinning or an affine-square reweighing as in \Cref{lem:affine-cusp-reweigh-gathers,lem:affine-consequences,lem:consequence-of-affine-2}, perform the following rounding steps.
    \begin{enumerate}
        \item Write $m=\pE x$, $z=\pE y$, and $R=\pE xy^\top$.
        \item If $m\neq0$, set $u=m/\norm m_2$ and let $v$ be the ``best response,'' given by a top
        eigenvector of
        \[
            (u\otimes I)^\top\Pi(u\otimes I).
        \]
        Test the product vector $u\otimes v$.
        \item If $z\neq0$, set $v=z/\norm z_2$ and let $u$ be the ``best response,'' given by a top
        eigenvector of
        \[
            (I\otimes v)^\top\Pi(I\otimes v).
        \]
        Test the product vector $u\otimes v$.
        \item If $R\neq0$, compute top left and right singular vectors of
        $R$ and test the resulting product vector.
    \end{enumerate}
    \item[Output:] The pair with the largest value of
    $\norm{\Pi(u\otimes v)}_2^2$ among all candidates.
    \end{description}
  \end{algorithm}
\end{mdframed}

\subsection{The rounding success condition}

We now describe sufficient conditions for at least one of the final rounding
steps to succeed.

\begin{lemma}[Low-soundness rounding]
\label{lem:low-soundness-rouding-success}
Let $\pE$ satisfy~\eqref{eq:constraints}, and write
\[
    m=\pE x,\qquad z=\pE y,\qquad R=\pE xy^\top.
\]
The two best-response candidates have scores at least
\[
    \norm m_2^2-O(n^{-\beta})
    \qquad\text{and}\qquad
    \norm z_2^2-O(n^{-\beta}),
\]
respectively. If $R\neq0$, the top-singular-vector candidate has score at
least
\[
    \left(\frac{\sigma_1(R)-n^{-\beta}}{\norm R_F}\right)_+^2.
\]
\end{lemma}

\begin{proof}
Suppose first that $m\neq0$, let $u=m/\norm m_2$, and set
\[
    T_u=(u\otimes I)^\top\Pi(u\otimes I).
\]
Put $\rho=(I-\Pi)(x\otimes y)$.  Using
\[
    \iprod{u,x}\norm y_2^2
    =
    \iprod{\Pi(u\otimes y)}{x\otimes y}
    +\iprod{u\otimes y}{\rho},
\]
the norm constraints, the residual constraint, and Cauchy--Schwarz give
\[
    \pE\iprod{u,x}^2
    \leq \pE\norm{\Pi(u\otimes y)}_2^2+O(n^{-\beta}).
\]
Writing $Y=\pE yy^\top$, the first term on the right is
$\Tr(T_uY)$.  Since $Y\succeq0$ and $\Tr(Y)=1+O(n^{-\beta})$, it is at
most $\lambda_{\max}(T_u)+O(n^{-\beta})$.  Finally,
$\norm m_2^2\leq\pE\iprod{u,x}^2$.
The top eigenvector of $T_u$ is a best response and hence has score at least
$\norm m_2^2-O(n^{-\beta})$.  The same argument with $x$ and $y$ exchanged
proves the other best-response bound.

For the singular-vector candidate, let $u,v$ be top left and right singular
vectors of $R$.  Jensen's inequality and the residual constraint give
\[
    \norm{(I-\Pi)R}_F\leq n^{-\beta}.
\]
Identifying $R$ with a vector in $\mathbb R^n\otimes\mathbb R^n$, we have
\begin{align*}
    \sigma_1(R)
    &=\iprod{u\otimes v}{R}\\
    &\leq\iprod{\Pi(u\otimes v)}{R}+n^{-\beta}\\
    &\leq\norm{\Pi(u\otimes v)}_2\norm R_F+n^{-\beta}.
\end{align*}
Rearranging and using nonnegativity proves the claimed bound.
\end{proof}

\subsection{A low-soundness reweighing lemma}

	We now prove the key technical ingredient, a ``low-soundness'' reweighing lemma which will establish one of the desired win conditions in~\cref{lem:low-soundness-rouding-success}.

	For the purposes of clarity of exposition we will work with the exact constraint set instead of \eqref{eq:eps_constraints}. It is easy to see that the $1/\poly(n)$ error does not impact the correctness of the argument below.

	\begin{lemma}
        \label{lem:low-soundness-pinning-lemma}
        Let $0 < \theta < c$ for some sufficiently small constant $c$, and let $\pE$ be a pseudodistribution of degree $\Omega\left( \theta\sqrt{n} \right)$ satisfying \eqref{eq:constraints}. Then, there exists a sequence of $O(\theta\sqrt{n})$ operations, each of which is either a pinning (in the sense of the conditioning on the value of some coordinate), or an affine reweigh (for some efficiently computable affine function $p = p(x,y)$ with $\pE p^2 > 0$, it reweighs the distribution by $p^2$), such that the resulting pseudodistribution $\pE'$ satisfies the following.
        \begin{enumerate}[label=(\alph*)]
            \item $\|\pE' x\|$ and $\|\pE' y\|$ are both bounded from below by $c \frac{\theta}{\sqrt{n}}$.
            \item $\|\pE' xy^\top \|_F \gtrsim \frac{\theta}{\sqrt{n}}$.
            \item one of the following win conditions is satisfied:
            \begin{itemize}
                \item $\|\pE x\| \gtrsim \theta$,
                \item $\|\pE y\| \gtrsim \theta$, or
                \item denoting $R = \pE xy^\top$, we have $\sigma_1(R) \gtrsim \theta \|R\|_{F}$.
            \end{itemize}
        \end{enumerate}
	\end{lemma}

		As in the previous section, we will start by performing the mean anchoring of \Cref{lem:mean-anchoring}.
		Define the potential
		\[ \Phi(u,v) = 2 \left( \|u\|^2 + \|v\|^2 + \|u\| \cdot \|v\| \right) \mper \]
		Our goal will be to show that if none of the win conditions above are satisfied, then we may reweigh the distribution to increase the potential $\Phi(m_x , m_y)$.

		As we shall see shortly, if $\sigma_1(R) \gtrsim \|\pE x\| \|\pE y\|$, then we may pin a coordinate as in \Cref{lem:potential-growth-43} to grow the potential.
		To deal with the other scenario, we shall use the following lemma.

		\begin{restatable}{lemma}{affinecuspbs}
			\label{lem:affine-cusp-bs-lemma}
			Let $\pE$ be a pseudodistribution on variables $(x,y)$ satisfying the constraints \eqref{eq:constraints}. Further suppose that all coordinates in some sets of indices $I_{x}, I_{y} \subseteq [n]$ have already been pinned such that they contribute squared $\ell_2$ mass of at least $Q$. Denoting $R = \pE xy^\top$, suppose that
			\[ \sigma_1(R) \le \frac{1}{2} \cdot \|\pE x\| \|\pE y\| \mper \]
			Then, there exists an affine function $p$ such that on reweighing the pseudodistribution by $p^2$, if $\wt{m}_x$ and $\wt{m}_y$ are the updated means, then
			\begin{equation}
				\label{eq:bsbsbs}
				\Phi\left( \wt{m}_x , \wt{m}_y \right) - \Phi\left( m_x , m_y \right) \gtrsim Q \mper
			\end{equation}
		\end{restatable}

		Let us put the pieces together and establish \Cref{lem:low-soundness-pinning-lemma}.

		\begin{proof}[Proof of \Cref{lem:low-soundness-pinning-lemma}]
			Start by pinning coordinates as in \Cref{lem:mean-anchoring}, for some $Q$ to be decided. Let the means of this distribution be $m_x^0$ and $m_y^0$. As in \Cref{lem:asymmetric-pinning}, while none of the win conditions are satisfied, we shall sequentially reweigh/pin the distribution. Let $m_x^t$ and $m_y^t$ be the mean after $t$ steps of this, and denote $\Phi_t = \Phi(m_x^t , m_y^t)$.

			If
			\[ \sigma_1(R) \ge \frac{1}{2} \cdot \|\pE x\| \cdot \|\pE y\| \mcom \]
			then the failure of the final win condition implies that
			\[ \|R_{F}\| \gtrsim \frac{1}{\theta} \|\pE x\| \cdot \|\pE y\| \mper \]
			When $\theta$ is sufficiently small, we may relate this to the Frobenius norm of the cross-covariance as
			\[ \|C\|_{F} \ge \|R\|_{F} - \|\pE x\| \cdot \|\pE y\| \gtrsim \frac{1}{\theta} \|\pE x\| \|\pE y\| \mper \]
			We then use \Cref{lem:potential-growth-43} to grow $\Phi$ by $\Omega\left( \theta^{-2} Q^{2/3} \Phi^{4/3} \right)$.
			
			On the other hand, if
			\[ \sigma_1(R) \le \frac{1}{2} \cdot \|\pE x\| \cdot \|\pE y\| \mcom \]
			then use \Cref{lem:affine-cusp-bs-lemma} so
			\[ \Phi_{t+1} \ge \Phi_{t} + \Omega(Q) \mper \]

			To summarize, we have that while none of the win conditions are satisfied,
			\[ \Phi_{t+1} - \Phi_{t} \gtrsim \min\left\{ Q , \theta^{-2} Q^{2/3} \Phi_t^{4/3} \right\} \mper \]
			Now, by the same proof as \Cref{lem:asymmetric-pinning}, the number of increments of the second type is at most $\frac{Q^{-1/3}}{\theta^{-2} Q^{2/3}} = O\left( \frac{\theta^2}{Q} \right)$.
			The first win condition additionally implies that $\Phi_{t} \lesssim \theta^2$ at all steps, so the number of increments of the first type is also $O\left( \frac{\theta^2}{Q} \right)$.

			Thus, we have performed $O(Qn)$ many rounds of pinning in the first phase, and $O\left( \frac{\theta^2}{Q} \right)$ many rounds of pinning/reweighing in the second. Balancing the two terms by picking $Q \asymp \frac{\theta}{\sqrt{n}}$, the desideratum follows.
		\end{proof}

		The two subsequent subsections are dedicated to establishing \Cref{lem:affine-cusp-bs-lemma}, and we encourage the reader to skip ahead to \Cref{subsec:low-soundness-proof}, where we use the above to prove our main theorem \Cref{thm:asym-bss-low-soundness}.

	\subsection{A low-soundness potential growth lemma}

		In this section, we prove the main technical input to the low-soundness result.

		\affinecuspbs*

		Assume without loss of generality that $\|m_x\| \le \|m_y\|$.
		
		The plan will be as follows. We start by determining the reweighings that change the mean $m_x$ as much as possible; this is the content of \Cref{lem:affine-cusp-reweigh-gathers}. However, we need to establish that growing the mean of $x$ does not suddenly make the mean of $y$ decay. This is where our assumption on the top singular value of $R$ will enter the picture; this is the content of \Cref{lem:corr-bound}, and will establish either the norm of $m_x$ or that of $m_y$ must grow. Unfortunately, this plan will not work out exactly as stated, as there will be a (possibly non-negligible) error term in these calculations.

		Define the normalized means $u = \frac{m_x}{\|m_x\|}$ and $v = \frac{m_y}{\|m_y\|}$. In the proof, we will bound
		\[ \Phi(\wt{m}_x , \wt{m}_y) - \Phi(m_x , m_y) \ge \psi\left( \langle \wt{m}_x , u \rangle , \langle \wt{m}_y , v \rangle \right) - \psi( \|m_x\| , \|m_y\| ) \mcom \]
		for $\psi(w_1,w_2) \defeq 2(w_1^2 + w_2^2 + |w_1w_2|)$. That is, we only use the contribution to the potential growth coming from changes along the mean.

		Let us also introduce the symbolic random variables $X = \langle x - \pE x , u \rangle$, $Y = \langle y - \pE y , v \rangle$, their variance $\sigma_x^2 = \wt{\Var} X$ and $\sigma_y^2 = \wt{\Var} Y$, as well as their normalized versions $\ol{X} = \frac{X}{\sigma_x}$ and $\ol{Y} = \frac{Y}{\sigma_y}$.

		The affine function we construct will be of the form $a \ol{X} + b$. In the following lemma, we compute how the means change on reweighing by a function of this form.
		
		\begin{lemma}
			\label{lem:affine-cusp-reweigh-gathers}
			Consider the $2 \times 2$ matrix
			\[ H = \begin{pmatrix} \pE \ol{X}^3 & 1 \\ 1 & 0 \end{pmatrix} \mper \]
			Define $\lambda > 0$ by $\lambda - \frac{1}{\lambda} = \pE \ol{X}^3$, so $\lambda_1 = \lambda$ and $\lambda_2 = - \frac{1}{\lambda}$ are the two eigenvalues of $H$.
			Then, there are reweighings $p_1 = a_1 \ol{X} + b_1$ and $p_2 = a_2 \ol{X} + b_2$, with coefficients equal to the corresponding eigenvectors of $H$, such that if the updated means are $m_x^1$ and $m_x^2$, then
			\[ \frac{ \langle m^i_x - m_x , u \rangle }{ \sigma_x } = \lambda_i \mper \]
			Furthermore, under these reweighings,
			\[ \frac{\langle m^i_y - m_y , v \rangle}{ \sigma_y } = \pE\left[\ol{Y} \cdot \ol{X}\right] \cdot \lambda_i + \left( \frac{\lambda_i^2}{1 + \lambda_i^2} \right) \cdot \xi \mcom \]
			where
			\[ \xi = \pE\left[ \ol{Y} \cdot \ol{X}^2 \right] - \pE \left[ \ol{Y} \cdot \ol{X} \right] \cdot \pE\left[ \ol{X}^3 \right] \mper \]
		\end{lemma}
		This quantity $\xi$, which we will analyse in much more detail shortly, represents the component of $\ol{Y}$ that is not predicted by $\ol{X}$ in a linear fashion.
		\begin{proof}
			This is a routine calculation. We reproduce the details here for completeness.
			\begin{align}
				\frac{ \langle \wt{m}_x - m_x , u \rangle }{ \sigma_x } &= \frac{\pE \ol{X} \left( a \ol{X} + b \right)^2 }{ \pE \left(a \ol{X} + b\right)^2 } \nonumber \\
					&= \frac{ a^2 \pE \ol{X}^3 + 2ab }{ a^2 + b^2 } \label{eq:x-movement} \mper
			\end{align}
			It is easy to interpret the numerator as the quadratic form $\begin{pmatrix} a \\ b \end{pmatrix}^\top H \begin{pmatrix} a \\ b \end{pmatrix}$.
			We similarly have, for $(a_i,b_i)$ being an eigenvector of the matrix with eigenvalue $\lambda_i$,
			\begin{align*}
				\frac{ \langle \wt{m}_y - m_y , v \rangle }{ \sigma_y } &= \frac{ \pE \left[a^2 \ol{Y} \ol{X}^2 + 2ab\ol{X} \ol{Y}\right] }{a^2 + b^2} \\
					&= \pE[\ol X\ol Y] \cdot \frac{a_i^2\pE\ol X^3+2a_i b_i}{a_i^2+b_i^2} + \frac{a_i^2}{a_i^2+b_i^2}\xi \\
					&= \pE[\ol X\ol Y] \cdot \lambda_i + \left( \frac{\lambda_i^2}{1 + \lambda_i^2} \right) \cdot \xi \mper
			\end{align*}
		\end{proof}

		The quantity $\E\left[ \ol{Y} \cdot \ol{X} \right]$ arising above can easily be bounded from below.

		\begin{lemma}
			\label{lem:corr-bound}
			Under the assumptions of \Cref{lem:affine-cusp-bs-lemma},
			\[ \pE\left[ XY \right] \le -\frac{1}{2} \|m_x\| \cdot \|m_y\| \mper \]
		\end{lemma}
		\begin{proof}
			By the definition of the singular value, the assumption in \Cref{lem:affine-cusp-bs-lemma} implies that
			\[ u^\top R v \le \frac{1}{2} \| m_x \| \cdot \|m_y\| \mcom \]
			so
			\[ \pE\left[ \langle x , u\rangle \cdot \langle y , v\rangle \right] \le \frac{1}{2} \| m_x \| \cdot \| m_y\| \mper \]
			Consequently,
			\begin{align*}
				\pE \left[ XY \right] &= \pE\left[ \langle x,u\rangle \cdot \langle y,v\rangle \right] - \|m_x\| \|m_y\| \\
					&\le -\frac{1}{2} \|m_x\| \|m_y\| \mper \qedhere
			\end{align*}
		\end{proof}

		If $\xi = 0$ in \Cref{lem:affine-cusp-reweigh-gathers}, then we are done. Indeed, consider the random reweigh which chooses $p_1$ with probability $\frac{1}{1+\lambda^2}$ and $p_2$ with probability $\frac{\lambda^2}{1+\lambda^2}$---this choice makes $\E \langle \wt{m}_y - m_y , u \rangle = \E \langle \wt{m}_x - m_x , u \rangle = 0$. We may then use a one-dimensional version of \Cref{lem:noise-helps} to show that
		\begin{align}
			\E \Phi\left( \wt{m}_x , \wt{m}_y \right) - \Phi\left( m_x , m_y \right) &\ge \sigma_x^2 + \sigma_y^2 \cdot \pE\left[ \ol{Y} \cdot \ol{X} \right]^2 \nonumber \\
				&\ge \sigma_x^2 + \frac{ \pE\left[ XY \right]^2 }{ \sigma_x^2 } \nonumber \\
				&\gtrsim \left| \pE\left[ XY \right]  \right| \nonumber \\
				&\gtrsim \|m_x\| \cdot \|m_y\| \gtrsim Q \mper \label{eq:xi-0-proof}
		\end{align}

		Let us reintroduce the ``residual'' $\xi$, with the same random choice of reweigh from above, showing that if $\xi$ is small, then the desired result follows.

		\begin{lemma}
			Suppose that $|\xi| \ll \frac{\|m_x\|}{\sigma_y}$. Then, one of the affine reweighings in \Cref{lem:affine-cusp-reweigh-gathers} satisfies \eqref{eq:bsbsbs}.
		\end{lemma}
		\begin{proof}
			It is no longer true, as in the manipulations in \eqref{eq:xi-0-proof}, that $\E \langle \wt{m}_y - m_y , v\rangle = 0$, but we nevertheless have above collection of equations with the extra error terms
			\begin{align}
				&\E \Phi\left( \wt{m}_x , \wt{m}_y \right) - \Phi\left( m_x , m_y \right) \nonumber \\
				&\qquad\ge \sigma_x^2 + \sigma_y^2 \cdot \pE[\ol{X}\ol{Y}]^2 \nonumber \\
				&\qquad\qquad + 2\xi \cdot \frac{\lambda^2}{(1+\lambda^2)^2} \left( 2 (2\|m_y\| + \|m_x\|) \sigma_y + \left( \lambda - \frac{1}{\lambda} \right) \left( \sigma_x \sigma_y + 2 \pE[\ol{X}\ol{Y}] \sigma_y^2 \right) \right) \nonumber \\
				&\qquad\qquad + 2 \xi^2 \cdot \frac{\lambda^2}{(1+\lambda^2)^2} \cdot \sigma_y^2 \nonumber \\
				&\qquad\gtrsim \sigma_x^2 + \sigma_y^2 \cdot \pE[\ol{X}\ol{Y}]^2 \\
				&\qquad\qquad + 2\xi \cdot \frac{\lambda^2}{(1+\lambda^2)^2} \left( 2 (2\|m_y\| + \|m_x\|) \sigma_y + \left( \lambda - \frac{1}{\lambda} \right) \left( \sigma_x \sigma_y + 2 \pE[\ol{X}\ol{Y}] \sigma_y^2 \right) \right) \mper \label{eq:monstrosity}
			\end{align}
			Let us condense the second term down. First, we have $\frac{\lambda^2}{(1+\lambda^2)^2} \le \frac{1}{4}$, and $\left| \frac{\lambda^2}{(1+\lambda^2)^2} \cdot \left( \lambda - \frac{1}{\lambda} \right) \right| \le \frac{1}{2}$. Thus,
			\begin{align*}
				&\left| 2\xi \cdot \frac{\lambda^2}{(1+\lambda^2)^2} \left( 2 (2\|m_y\| + \|m_x\|) \sigma_y + \left( \lambda - \frac{1}{\lambda} \right) \left( \sigma_x \sigma_y + 2 |\pE[\ol{Y} \cdot \ol{X}]| \sigma_y^2 \right) \right) \right| \\
				&\qquad\qquad\lesssim |\xi| \cdot \left( (2\|m_y\| + \|m_x\|) \sigma_y + \left( \sigma_x \sigma_y + 2 |\pE[\ol{Y} \cdot \ol{X}]| \sigma_y^2 \right) \right) \\
				&\qquad\qquad\lesssim |\xi| \cdot \left( \|m_y\| \sigma_y + \left( \sigma_x \sigma_y + 2 \pE[\ol{Y} \cdot \ol{X}] \sigma_y^2 \right) \right) \mper
			\end{align*}
			We next show that the parenthesised terms may be discarded. First, \Cref{lem:corr-bound} yields that 
			\[ \|m_x\| \cdot \|m_y\| \lesssim |\E[XY]| = \sigma_x \sigma_y \cdot |\E[\ol{X}\ol{Y}]| \lesssim \sigma_x^2 + \E[\ol{X}\ol{Y}]^2 \sigma_y^2 \mper \]
			We may now apply the AM-GM inequality on the parenthesised terms so for some small constant $\eps$,
			\[ |\xi| \cdot \left(\sigma_x \sigma_y + \sigma_y^2 \left| \pE[\ol{Y}\ol{X}] \right|\right) \lesssim \eps \left(\sigma_x^2 + \E[\ol{X}\ol{Y}]^2 \sigma_y^2\right) + \frac{1}{\eps} \xi^2 \sigma_y^2 \mper \]
			Therefore, if $|\xi| \sigma_y \ll \|m_x\|$, then we may simplify the above as
			\begin{align*}
				&\E \Phi\left( \wt{m}_x , \wt{m}_y \right) - \Phi\left( m_x , m_y \right) \\
					&\ge \underbrace{ \left(\sigma_x^2 + \sigma_y^2 \cdot \pE[\ol{X}\ol{Y}]^2\right)\left( 1 - C \eps \right) }_{\gtrsim \|m_x\| \cdot \|m_y\|} - \underbrace{|\xi| \sigma_y \|m_y\|}_{\ll \|m_x\| \|m_y\|} - \underbrace{\xi^2 \sigma_y^2}_{\ll \|m_x\|^2} \\
					&\gtrsim \|m_x\| \|m_y\| \ge Q \mper
			\end{align*}
			Here, the second inequality uses the same manipulations from \Cref{eq:xi-0-proof}.
		\end{proof}

		The final part of the proof, which we execute in \Cref{subsec:cusp}, will establish that if the $\xi$ term cannot be absorbed into the first, then it is easy to get a gain of $Q$.
		
		\begin{restatable}{lemma}{largexilemma}
			\label{lem:home-stretch-affine-cusp}
			Suppose that $|\xi| \gtrsim \frac{\|m_x\|}{\sigma_y}$. Then, there is an affine reweighing satisfying \eqref{eq:bsbsbs}.
		\end{restatable}

		Before doing this, we exhibit numerous consequences of the failure of every affine reweigh.

	\subsection{Consequences of affine failure}

		Recall
		\[ \xi = \pE\left[ \ol{Y} \ol{X}^2 \right] - \pE\left[ \ol{Y} \ol{X} \right] \cdot \pE\left[ \ol{X}^3 \right] \mper \]
		Let us express $\ol{Y} = \rho \ol{X} + \sqrt{1-\rho^2} W$, where $\rho = \E[ \ol{X} \ol{Y} ]$, so $\E[\ol{X} W] = 0$. We then have the far simpler expression
		\begin{equation}
			\label{eq:xi-expression}
			\xi = \sqrt{1-\rho^2} \pE\left[ \ol{X}^2 W \right] \mper
		\end{equation}
		
		In this section, we establish the following consequences of the failure of every affine function.

		\begin{lemma}
			\label{lem:affine-consequences}
			Suppose that there is no affine reweigh satisfying \eqref{eq:bsbsbs}, in the sense that for some small constant $\delta$, for any affine function $p$, if $\wt{m}_x$ and $\wt{m}_y$ are the updated means, then
			\begin{equation}
				\label{eq:quantitative-no-progress}
				\Phi(\wt{m}_x , \wt{m}_y) - \Phi(m_x , m_y) < \delta Q \mper
			\end{equation}
			Further assume that
			\[ |\xi| \gtrsim \frac{\|m_x\|}{\sigma_y} \]
			Then,
			\begin{enumerate}[label=(\roman*)]
				\item $\displaystyle \frac{\|m_x\|}{\|m_y\|} \lesssim \delta$,
				\item $\rho = \left( - \frac{1}{2} + O(\delta) \right) \frac{\sigma_x}{\sigma_y}$,
				\item $\displaystyle\sigma_x \lesssim \sqrt{\delta} \cdot \|m_y\|$,
				\item $\displaystyle |\xi| \lesssim \sqrt{\delta} \cdot \frac{\sigma_x}{\sigma_y}$, and
				\item $\sigma_x^2 \gtrsim \|m_x\| \|m_y\| \mper$
			\end{enumerate}
		\end{lemma}

		Before proceeding, we introduce the following lemma which describes how much an affine reweigh can shift an affine function by.

		\begin{lemma}
			\label{lem:best-affine-reweigh-for-linear-fn}
			For any linear function $\ell$, we have
			\[
				\max_{\substack{\text{affine reweigh $p$} \\ \pE p^2 = 1}} \left| \pE\left[ p^2 \ell \right]  \right| \ge \sqrt{ \pE \ell^2 } \mper
			\]
			Furthermore, there is an efficiently computable affine function $p$ witnessing this inequality.
		\end{lemma}
		\begin{proof}
			Express $\ell = h \sigma + \mu$, with $\pE h = 0$ and $\pE h^2 = 1$. We restrict our attention to $p$ of the form $p = ah + b$. We then have
			\begin{align*}
				\max_{p} \left| \pE\left[ p^2 \ell \right]  \right| &\ge \max_{(a,b)} \left| \begin{pmatrix} a \\ b \end{pmatrix}^\top \begin{pmatrix} \mu & \sigma \\ \sigma & \pE[h^2 \ell] \end{pmatrix} \begin{pmatrix} a \\ b \end{pmatrix} \right| \\
					&= \left\| \begin{pmatrix} \mu & \sigma \\ \sigma & \pE[h^2 \ell] \end{pmatrix} \right\|_{\mathrm{op}} \\
					&\ge \left\| \begin{pmatrix} \mu & \sigma \\ \sigma & \pE[h^2 \ell] \end{pmatrix} \begin{pmatrix} 1 \\ 0 \end{pmatrix} \right\| \\
					&= \sqrt{ \mu^2 + \sigma^2 } = \sqrt{ \pE \ell^2 } \mper \qedhere
			\end{align*}
		\end{proof}

		\begin{lemma}
			\label{lem:consequence-of-affine-1}
			Under the assumptions of \Cref{lem:affine-consequences},
			\begin{enumerate}[label=(\roman*)]
				\item $\displaystyle\sigma_y \sqrt{ 1 - \rho^2 } \lesssim \sqrt{\delta Q}$ and
				\item $\displaystyle\sigma_x^2 \gtrsim \|m_x\| \cdot \|m_y\|\mper$	
			\end{enumerate}
			If either of these inequalities fails, a witness affine function that violates \eqref{eq:quantitative-no-progress} is the affine reweigh described by \Cref{lem:best-affine-reweigh-for-linear-fn} for the linear function $\left( 2\|m_x\| + \|m_y\| \right) X + \left( 2\|m_y\| + \|m_x\| \right) Y$.
		\end{lemma}
		\begin{proof}
			Consider the ``tangent'' linear function
			\[ \ell(a,b) = \left( 2\|m_x\| + \|m_y\| \right) a + \left( 2 \|m_y\| + \|m_x\| \right) b \mcom \]
			so $\pE \ell(\langle u,x\rangle , \langle v,y\rangle) = \Phi\left( m_x , m_y \right)$. Then, Cauchy--Schwarz implies that for any $(a,b)$, shortening $\Phi = \Phi(m_x,m_y)$,
			\begin{equation}
				\label{eq:tangent-cauchy-schwarz}
				\ell(a,b)^2 \le \Phi \cdot \psi(a,b) \mper	
			\end{equation}
			Now, we have
			\begin{align*}
				\Phi \cdot \max_{\substack{\text{affine reweigh $p$} \\ \pE p^2 = 1}} \psi\left( \wt{m}_x , \wt{m}_y \right) &\stackrel{\eqref{eq:tangent-cauchy-schwarz}}{\ge} \max_{\substack{\text{affine reweigh $p$} \\ \pE p^2 = 1}} \left(\pE\left[ p^2 \ell\left( \langle x,u\rangle , \langle y,v\rangle \right) \right]\right)^2 \\
					&\ge \pE\left[ \ell\left( \langle x,u\rangle , \langle y,v\rangle \right)^2 \right] \\
					&= \Phi^2 + \wt{\Var} \ell\left( \langle x,u\rangle , \langle y,v\rangle \right) \mcom
			\end{align*}
			where the second inequality is \Cref{lem:best-affine-reweigh-for-linear-fn}.
			It follows that 
			\[ \wt{\Var} \ell\left( \langle x,u\rangle , \langle y,v\rangle \right) \le \delta Q \Phi \mper \]
			To avoid having to deal with a cross-covariance term, we can write
			\begin{align*}
				\ell\left( X, Y \right) &= \ell ( X , \sigma_y \cdot \left(\frac{\rho}{\sigma_x} X + \sqrt{1-\rho^2} W\right) ) \mcom
			\end{align*}
			so
			\begin{align*}
				\wt{\Var} \ell\left( X,Y \right) &\gtrsim \|m_y\|^2 \cdot \sigma_y^2 \cdot (1-\rho^2) \mcom
			\end{align*}
			and thus
			\[ \sigma_y \sqrt{1-\rho^2} \lesssim \sqrt{\delta Q} \mcom \]
			establishing the first inequality. For the second inequality, we have
			\begin{align}
				\delta Q \Phi &\gtrsim \Var \ell(X,Y) \nonumber \\
					&\gtrsim \frac{1}{\sigma_x^2} \left(\E\left[ \ell(X,Y) \cdot X \right]\right)^2 \nonumber \\
				\sqrt{\delta Q} &\gtrsim \frac{1}{\sigma_x} \left| \frac{2\|m_x\| + \|m_y\|}{\|m_x\| + 2\|m_y\|} \cdot \sigma_x^2 + \E[XY] \right| \mper \label{eq:main-cusp-equation}
			\end{align}
			Now recall \Cref{lem:corr-bound}. This implies that
			\[ \|m_x\| \cdot \|m_y\| \lesssim -\E\left[XY\right] \lesssim \sigma_x^2 + \sigma_x \sqrt{\delta Q} \lesssim \sigma_x^2 + \delta Q \mcom \]
			and since the $\delta Q$ term is negligible in comparison to the left-hand side, we have the desired inequality.
		\end{proof}

		We also have the following bound on a third moment quantity that shows up in the definition of $\xi$.

		\begin{lemma}
			\label{lem:consequence-of-affine-2}
			Under the assumptions of \Cref{lem:affine-consequences},
			\[ \left| \pE[\ol{X}^2 W] \right| \lesssim \frac{\|m_y\|}{\sigma_x} \mper \]
			If this inequality fails, a witness affine function that violates \eqref{eq:quantitative-no-progress} is one of the two choices $\frac{\ol{X} \pm W}{\sqrt{2}}$.
		\end{lemma}
		\begin{proof}
			Consider the two reweighs
			\[ p_{\pm} = \frac{\ol{X} \pm W}{\sqrt{2}} \mper \]
			Due to the orthonormality of $\ol{X}$ and $W$, we have $\pE[p_{\pm}^2] = 1$. Furthermore, their updated means $\wt{m}_x^{\pm}$ satisfy
			\[ \langle \wt{m}_x^+ - \wt{m}_x^- , u\rangle = \sigma_x \pE\left[ \left(p_+^2 - p_-^2\right) \ol{X} \right] = 2\sigma_x \pE\left[ \ol{X}^2 W \right] \mper \]
			We further have
			\[ \Phi(\wt{m}_x^\pm , \wt{m}_y^{\pm}) \le \Phi(m_x , m_y) + \delta Q \lesssim \|m_y\|^2 \mcom \]
			where the final inequality is simply because $\Phi(m_x , m_y) \asymp \|m_y\|^2$. We can lower bound the quantity on the left as
			\[ \Phi(\wt{m}_x^\pm , \wt{m}_y^{\pm}) \ge \langle \wt{m}_x^{\pm} , u \rangle^2 \mper \]
			It follows that
			\[ \sigma_x \left| \pE\left[ \ol{X}^2 W \right] \right| \lesssim \|m_y\| \mcom \]
			and rearranging completes the proof.
		\end{proof}

		\begin{lemma}
			\label{lem:consequence-of-affine-3}
			Under the assumptions of \Cref{lem:affine-consequences},
			\[ \frac{\|m_x\|}{\|m_y\|} \lesssim \delta \]
			and
			\[ \E\left[XY\right] = \left( - \frac{1}{2} + O(\delta) \right) \sigma_x^2 \mper \]
		\end{lemma}
		\begin{proof}
			Combining the assumption and \Cref{lem:consequence-of-affine-2}, along with the expression for $\xi$ in \Cref{eq:xi-expression}, we have
			\[ \sqrt{1-\rho^2} \cdot \frac{\|m_y\|}{\sigma_x} \gtrsim \frac{\|m_x\|}{\sigma_y} \mper \]
			Substituting \Cref{lem:consequence-of-affine-1}(i), we have
			\begin{equation}
				\label{eq:another-conseq-4}
				\frac{\|m_x\|}{\|m_y\|} \le \frac{\sqrt{\delta Q}}{\sigma_x}
			\end{equation}
			That is, the $x$-mean is significantly smaller than the $y$-mean. On the other hand, \Cref{lem:consequence-of-affine-1}(ii) yields that
			\[ \sigma_x^2 \gtrsim \|m_x\| \cdot \|m_y\| \mper \]
			Substituting this back in above yields
			\[ \frac{\|m_x\|^{3/2}}{\|m_y\|^{1/2}} \lesssim \sqrt{\delta Q} \mcom \]
			and bounding $\|m_x\| \ge \sqrt{Q}$ yields
			\[ \frac{\|m_x\|}{\|m_y\|} \lesssim \delta \mper \]
			For the second part, we plug this back into \eqref{eq:main-cusp-equation}, which asserts that
			\[ \sqrt{\delta Q} \gtrsim \frac{1}{\sigma_x} \left| \frac{2\|m_x\| + \|m_y\|}{\|m_x\| + 2\|m_y\|} \cdot \sigma_x^2 + \E[XY] \right| \mper \]
			Indeed, the first term within the absolute value is now equal to $\left( \frac{1}{2} + O(\delta) \right) \cdot \sigma_x^2$, so
			\[ \E[XY] = \sigma_x^2 \left(- \frac{1}{2} + O\left(\delta\right) + O\left(\frac{\sqrt{\delta Q}}{\sigma_x}\right)\right) \]
			To conclude, we have
			\[ \frac{\sqrt{Q}}{\sigma_x} \le \frac{\|m_x\|}{\sigma_x} \lesssim \sqrt{ \frac{\|m_x\|}{\|m_y\|} } \lesssim \sqrt{\delta} \mper \qedhere \]
		\end{proof}

		Finally, let us put the above lemmas together to establish \Cref{lem:affine-consequences}.

		\begin{proof}[Proof of \Cref{lem:affine-consequences}]
			(i) and (ii) are \Cref{lem:consequence-of-affine-3}. (iii) follows from \eqref{eq:another-conseq-4} and bounding $\sqrt{Q} \le \|m_x\|$. (v) is \Cref{lem:consequence-of-affine-1}(ii). For (iv), we have by \Cref{lem:consequence-of-affine-2} that
			\[ |\xi| = \sqrt{1-\rho^2} \E[\ol{X}^2 W] \lesssim \sqrt{1-\rho^2} \cdot \frac{\|m_y\|}{\sigma_x} \mper \]
			Now substitute \Cref{lem:consequence-of-affine-1}(i), we may bound the above further as
			\[ |\xi| \lesssim \sqrt{\delta Q} \cdot \frac{\|m_y\|}{\sigma_x \sigma_y} \le \sqrt{\delta} \cdot \frac{\|m_x\| \|m_y\|}{\sigma_x \sigma_y} \lesssim \sqrt{\delta} \cdot \frac{\sigma_x}{\sigma_y} \mcom \]
			where the final inequality is (v).
		\end{proof}

	\subsection{\texorpdfstring{The proof when $|\xi|$ is large}{The proof when |ξ| is large}}
		\label{subsec:cusp}

		Let us complete the proof.

		\largexilemma*

		\begin{proof}
			Recall the two reweighs from \Cref{lem:affine-cusp-reweigh-gathers}, corresponding to the evaluation of $\psi$ at the two points
			\[ z_i = \left( \|m_x\| + \sigma_x \lambda_i , \|m_y\| + \rho \sigma_y \lambda_i + \sigma_y \left( \frac{\lambda_i^2}{1 + \lambda_i^2} \right) \xi \right) \mcom \]
			and $z_0 = \left( \|m_x\| , \|m_y\| \right)$, where the $\lambda_i$ are defined as follows. We start by defining $\lambda$ by $\lambda - \frac{1}{\lambda} = \pE \ol{X}^3$, then set $\lambda_1 = \lambda$ and $\lambda_2 = - \frac{1}{\lambda}$.

			For starters, observe that the failure of these reweighs implies that
			\begin{equation}
				\label{eq:another-progress-condition}
				\sigma_x |\lambda_i| \lesssim \|m_y\| \mper
			\end{equation}
			Indeed,
			\[ \sigma_x^2 \lambda_i^2 \lesssim \|m_x\|^2 + \psi(z_i) \le \|m_x\|^2 + \psi(z_0) + \delta Q \lesssim \|m_y\|^2 \mper \]

			Let us compute the value of the potential at each of these two values. Define
			\[ \psi_+(z) = 2 \left( z_1^2 + z_2^2 + z_1z_2 \right) \text{ and } \psi_-(z) = 2 \left( z_1^2 + z_2^2 - z_1z_2 \right) \mper \]
			Let us start with $\psi_+(z_1) - \psi(z_0)$. Write $\omega_1 = \left( \frac{\lambda_1^2}{1 + \lambda_1^2} \right) \le 1$. Then, denoting $\beta = \frac{\rho \sigma_y}{\sigma_x} = -\frac{1}{2} + O(\delta)$,
			\begin{align}
				\frac{1}{2} \left(\psi_+(z_1) - \psi(z_0)\right) &= \sigma_x^2 \lambda_1^2 + \sigma_y^2 \left( \rho \lambda_1 + \omega_1 \xi \right)^2 + 2 \sigma_x \|m_x\| \lambda_1 + 2 \sigma_y \|m_y\| \left( \rho \lambda_1 + \omega_1 \xi \right) \nonumber\\
					&\qquad\qquad + \left(\|m_x\| + \sigma_x\lambda_1\right) \left( \rho \sigma_y \lambda_1 + \omega_1 \sigma_y \xi \right) + \|m_y\| \cdot \sigma_x \lambda_1 \nonumber\\
					&= \sigma_x^2 \lambda_1^2 \left( 1 + \beta^2 + \beta \right) \nonumber\\
					&\qquad + \sigma_x \lambda_1 \underbrace{ \left( 2\|m_x\| + \|m_y\| + \beta \left( \|m_x\| + 2\|m_y\| \right) \right) }_{\lesssim \delta \|m_y\| \text{ by \Cref{lem:affine-consequences}(i,ii)}} \nonumber\\
					&\qquad + \omega_1 \xi \sigma_y \left( \underbrace{\left( 2 \|m_y\| + \|m_x\| \right)}_{\gg \lambda_1 \sigma_x \text{ by \Cref{lem:affine-consequences}(iii)}} + \lambda_1 \sigma_x \underbrace{\left(2\beta + 1\right)}_{O(\delta)\text{ by \Cref{lem:affine-consequences}(ii)}} \right) \nonumber\\
					&\qquad + \underbrace{\omega_1^2 \xi^2 \sigma_y^2}_{\ge 0} \nonumber\\
					&\ge \frac{1}{2} \cdot \sigma_x^2\lambda_1^2 - C \omega_1 \sigma_y \|m_y\| \cdot |\xi| \\
					&\ge \frac{1}{2} \cdot \sigma_x^2\lambda_1^2 - C \delta \cdot \omega_1 \|m_y\|^2 \label{eq:positive-branch} \mcom
			\end{align}
			where the second-to-last inequality absorbs the second term into the first, and the final inequality is \Cref{lem:affine-consequences}(iii,iv).
			Since the above increment is bounded from above by $\delta Q$, we have
			\begin{equation}
				\label{eq:dplus-small}
				\sigma_x^2\lambda_1^2 \lesssim \delta Q + \delta \cdot \omega_1 \|m_y\|^2 \lesssim \delta \|m_y\|^2 \mper
			\end{equation}
			Similarly, for the other branch, we may compute $\psi_-(z_2) - \psi(z_0)$.
			Again, write $\omega_2 = \left( \frac{\lambda_2^2}{1 + \lambda_2^2} \right)$, and denote $\beta = \frac{\rho \sigma_y}{\sigma_x} = -\frac{1}{2} + O(\delta)$, so
			\begin{align}
				\frac{1}{2} \left( \psi_-(z_2) - \psi(z_0) \right) &= - 2 \|m_x\| \|m_y\| + \sigma_x^2 \lambda_2^2 + \sigma_y^2 \left( \rho \lambda_2 + \omega_2 \xi \right)^2 + 2 \sigma_x \|m_x\| \lambda_2 + 2 \sigma_y \|m_y\| \left( \rho \lambda_2 + \omega_2 \xi \right) \nonumber \\
					&\qquad\qquad - \left(\|m_x\| + \sigma_x\lambda_2\right) \left( \rho \sigma_y \lambda_2 + \omega_2 \sigma_y \xi \right) - \|m_y\| \cdot \sigma_x \lambda_2 \nonumber \\
						&= - 2 \|m_x\| \|m_y\| \nonumber \\
						&\qquad + \underbrace{\sigma_x^2 \lambda_2^2 \left( 1 + \beta^2 - \beta \right)}_{\ge 0} \nonumber \\
						&\qquad + \underbrace{\sigma_x \lambda_2 \left( 2\|m_x\| - \|m_y\| + \beta \left( - \|m_x\| + 2\|m_y\| \right) \right)}_{\gtrsim \sigma_x |\lambda_2| \|m_y\|} \nonumber \\
						&\qquad +  \underbrace{ \omega_2 \xi \sigma_y\left( \left( 2 \|m_y\| - \|m_x\| \right) + \lambda_2 \sigma_x \left(2\beta - 1\right) \right) }_{\stackrel{\eqref{eq:another-progress-condition}}{\gtrsim} - |\omega_2 \xi| \sigma_y \|m_y\|} \nonumber \\
						&\qquad + \underbrace{\omega_2^2 \xi^2 \sigma_y^2}_{\ge 0} \nonumber \\
					&\ge \sigma_x |\lambda_2| \|m_y\| - 2 \|m_x\| \|m_y\| - C |\omega_2 \xi| \sigma_y \|m_y\| \label{eq:negative-branch} \mper
			\end{align}
			Now, since the increment above in \Cref{eq:negative-branch} is bounded from above by $O(\|m_x\|^2) = O(\|m_x\| \|m_y\|)$, we have
			\begin{align*}
				\sigma_x |\lambda_2| &\lesssim \|m_x\| + |\omega_2 \xi| \sigma_y \\
					&\lesssim \|m_x\| + \sqrt{\delta} \cdot \sigma_x |\omega_2|  \\
					&= \|m_x\| + \sqrt{\delta} \cdot \sigma_x \cdot \frac{|\lambda_2|}{|\lambda_1| + |\lambda_2|} \\
					&\le \|m_x\| + \sqrt{\delta} \cdot \sigma_x \cdot \frac{|\lambda_2|}{2\sqrt{|\lambda_1\lambda_2|}} \\
					&= \|m_x\| + \sqrt{\delta} \cdot \frac{\sigma_x|\lambda_2|}{2} \\
				|\lambda_2| \sigma_x &\lesssim \|m_x\| \mper
			\end{align*}
			Now, $|\lambda_1| = \frac{1}{|\lambda_2|}$, and by \Cref{lem:affine-consequences}(v), we have $\sigma_x^2 \gtrsim \|m_x\| \cdot \|m_y\|$. Plugging these in above yield that
			\begin{equation}
				\label{eq:dplus-large}
				|\lambda_1| \sigma_x \gtrsim \|m_y\| \mper
			\end{equation}
			However, \eqref{eq:dplus-small} contradicts \eqref{eq:dplus-large}, completing the proof.
		\end{proof}

\subsection{Proof of the main theorem}
\label{subsec:low-soundness-proof}

We are now ready to prove the main result, restated for convenience.

\lowsoundnessthm*
\begin{proof}
Let $\beta = 2$ in $\mathcal{A}$. We first prove completeness.  Normalize a YES witness so that
$\norm x_2=\norm y_2=1$, and round both vectors to the grid to obtain
$\hat x,\hat y\in\Sigma^n$ with
\[
    \norm{x-\hat x}_2,\norm{y-\hat y}_2\leq c_0n^{-2}
\]
for a sufficiently small absolute constant $c_0$.  Then
\[
    \left|\norm{\hat x}_2^2-1\right|,
    \left|\norm{\hat y}_2^2-1\right|\leq n^{-2},
\]
and
\[
    \norm{\hat x\otimes\hat y-x\otimes y}_2
    \leq\norm{(\hat x-x)\otimes\hat y}_2
      +\norm{x\otimes(\hat y-y)}_2
    =O(n^{-2}).
\]
Since $(I-\Pi)(x\otimes y)=0$, decreasing $c_0$ if necessary shows that
the point mass at $(\hat x,\hat y)$ satisfies~\eqref{eq:constraints}.
Hence the relaxation is feasible in the YES case.

We next prove soundness.  Suppose that the relaxation is feasible, and set
\[
    \theta=K\sqrt{\frac qn}
\]
for a sufficiently large absolute constant $K$.  We first dispose of the
range in which $\theta$ is not smaller than the constant required by the
pinning lemma.  In that range $q=\Omega(n)$, and we apply
\cref{thm:asym-bss-close-to-one-main} with
$\eps=1-q/n$.  Since $\eps>1/2$, its running time is
$n^{O(\sqrt n)}=n^{O(\sqrt q)}$, and its soundness condition is precisely
the NO condition above.  We may therefore assume that $\theta$ is smaller
than the required constant and apply
\cref{lem:low-soundness-pinning-lemma}.  After
\[
    O(\theta\sqrt n)=O(\sqrt q)
\]
coordinate pinnings and affine-square reweightings, one of its three win
conditions holds.  Fix an absolute constant $c_{\mathrm{round}}>0$ smaller
than the implicit constants in these conditions.

If either mean is large, the corresponding best-response bound in
\cref{lem:low-soundness-rouding-success} gives
\[
    \norm{\Pi(u\otimes v)}_2^2
    \geq c_{\mathrm{round}}^2K^2\frac qn-O(n^{-2}).
\]
Here $n^{-2}\leq(q/n)/n$ because $q>1$.

It remains to consider the singular-value condition.
The final conclusion of~\cref{lem:low-soundness-pinning-lemma} gives
\[
    \norm R_F\geq\frac{\sqrt q}{n},
    \qquad
    \frac{n^{-2}}{\norm R_F}
    \leq\frac{1}{n\sqrt q}
    \leq\frac1{\sqrt n}\sqrt{\frac qn}.
\]
The singular-vector bound in
\cref{lem:low-soundness-rouding-success} therefore yields
\[
    \norm{\Pi(u\otimes v)}_2^2
    \geq
    \left(c_{\mathrm{round}}K-\frac1{\sqrt n}\right)^2\frac qn.
\]
Choosing $K$ sufficiently large makes both rounding bounds strictly larger
than $q/n$.  Thus, in every case,
\[
    \norm{\Pi(u\otimes v)}_2^2>\frac qn.
\]
Therefore, a feasible relaxation contradicts the NO condition.  The algorithm
enumerates all coordinate branches of the required depth, and the affine
reweightings and all final rounding candidates are computable from the
corresponding moment matrices.  It therefore finds such a pair whenever
the relaxation is feasible.

Finally, the grid has polynomial size, the relaxation has degree
$O(\sqrt q)$, and there are $n^{O(\sqrt q)}$ coordinate branches.  Both the feasibility test and
the rounding procedure consequently run in time
\[
    n^{O(\sqrt q)}.
\]
The efficient computability of the affine reweighs is immediate because each reweigh is one of five explicit reweighs: the two in \Cref{lem:affine-cusp-reweigh-gathers}, that in \Cref{lem:consequence-of-affine-1}, or the two in \Cref{lem:consequence-of-affine-2}.
\end{proof}

\section{Fine-grained pinning lemma}
\label{sec:fine-grained-pinning}
We prove the following fine-grained pinning lemma:

\begin{lemma}
\label{lem:fine-grained-pinning-loglog}
    Let $\mu_0$ be a distribution on $\mathbb{R}^n$ with finite second moment and let $0<\gamma \leq \sqrt n$. There is a distribution over $(S,x_S)$ where $S \subseteq [n]$ with $|S| \leq O(\sqrt n/\gamma)$ and $x_S \in \R^{|S|}$ such that
    \begin{gather*}
        \E_{S,x_S} \Norm{\E_{x \sim \mu_0 \mid x_S} x - \E_{x \sim \mu_0} x}_2^2 \ge \min\left\{ \frac{1}{\gamma} \cdot \E_{S,x_S} \Norm{\Cov(\mu_0 \, | \, x_S)}_F \mcom \frac{1}{\gamma^2} \cdot \Tr\left(\Cov(\mu_0)\right) \right\}
    \end{gather*}
    and furthermore $(S,x_S)$ give a decomposition of $\mu_0$, in that $\E_{S,x_S} (\mu_0 \, | \, x_S) = \mu_0$.
    The same conclusion holds for a pseudodistribution $\mu_0$ of degree at least $O(\sqrt n/\gamma)$.
\end{lemma}

We first record the estimates used in the proof. Throughout the claims, fix an initial distribution $\mu_0$ on $\mathbb{R}^n$ with finite second moment and $\Tr(\Cov(\mu_0))>0$. For a conditional distribution $\nu$ of $\mu_0$, let
\[
    m_\nu \vcentcolon= \E_{x \sim \nu} x \qquad \text{and} \qquad \Sigma_\nu \vcentcolon= \Cov(\nu)\,.
\]
Let $m_0 \vcentcolon= \E_{x \sim \mu_0}x$.
For a distribution $\pi_\ell$ over conditional distributions obtained from pinning coordinates of $\mu_0$, define
\[
    s_\ell \vcentcolon= \frac{\E_{\nu \sim \pi_\ell} \norm{m_\nu-m_0}_2^2}{\Tr(\Cov(\mu_0))}.
\]
As in other analyses of global correlation rounding, we will track the evolution of a potential function, in this case $s_\ell$, as we pin coordinates.
We construct our pinning of $\mu_0$ via the following process: given a current set of pinned coordinates $S$ pinned to values $x_S$, choose a coordinate $i$ on each branch with weights proportional to $\Var_\nu(x_i)$.
Then sample a value $x_i$ from the marginal distribution on coordinate $i$ conditioned on $x_S$, and update $S \leftarrow S \cup \{i \}$.
In the following claims, $\pi_\ell$ is the distribution over pinnings resulting from $\ell$ steps of this process.

\begin{claim}
\label{claim:fine-grained-step}
We have
\[
    s_{\ell+1}-s_\ell
    \geq
    \frac{1}{\Tr(\Cov(\mu_0))}\E_{\nu \sim \pi_\ell}
    \frac{\Norm{\Sigma_\nu}_F^2}{\Tr(\Sigma_\nu)}\,,
\]
where the ratio is interpreted as $0$ when $\Tr(\Sigma_\nu)=0$. Also,
\[
    s_\ell + \frac{\E_{\nu \sim \pi_\ell}\Tr(\Sigma_\nu)}{\Tr(\Cov(\mu_0))}=1\,.
\]
\end{claim}
We first state a simple and standard fact that will be necessary in the proof of the claim.
\begin{fact}[Variance drop from conditioning (see~\cite{BarakRS11})]
\label{fact:variance-drop-conditioning}
    For any two real-valued random variables $x,y$ with $\Var(y)>0$,
    \[
        \Var(x)-\E_y\Var(x\mid y)
        =
        \Var_y(\E[x\mid y])
        \geq
        \Omega(1) \cdot \frac{\Cov(x,y)^2}{\Var(y)}.
    \]
\end{fact}
We now return to the proof of~\cref{claim:fine-grained-step}.
\begin{proof}
    The covariance identity
    \[
        \Cov(\mu_0)
        =
        \E_{\nu \sim \pi_\ell}\Sigma_\nu
        +
        \E_{\nu \sim \pi_\ell}(m_\nu-m_0)(m_\nu-m_0)^\top
    \]
    follows because sampling $\nu \sim \pi_\ell$ and then $x \sim \nu$ gives a sample from $\mu_0$. Taking traces and dividing by $\Tr(\Cov(\mu_0))$ gives the second display.

    Fix a branch $\nu$ with $\Tr(\Sigma_\nu)>0$; trace-zero branches contribute $0$ and are left unchanged. By \cref{fact:variance-drop-conditioning}, conditioning on coordinate $i$ decreases the variance of $x_j$ by at least
    \[
        \Var_\nu(x_j)-\E_{x_i}\Var_{\nu \mid x_i}(x_j)
        \geq
        \frac{\Cov_\nu(x_i,x_j)^2}{\Var_\nu(x_i)}\,,
    \]
    where zero-variance coordinates are omitted. If $i$ is sampled with probability proportional to $\Var_\nu(x_i)$, then averaging over the choice of $i$ gives
    \[
        \Tr(\Sigma_\nu)-\E_{i,x_i}\Tr(\Sigma_{\nu \mid x_i})
        \geq
        \frac{\Norm{\Sigma_\nu}_F^2}{\Tr(\Sigma_\nu)}\,.
    \]
    We then pass to the conditional distribution $\nu \mid x_i$.

    It remains to translate the trace drop to mean drift. For each fixed coordinate $i$, since $\E_{x_i} m_{\nu \mid x_i}=m_\nu$,
    \begin{align*}
        \E\norm{m_{\nu \mid x_i}-m_0}_2^2-\norm{m_\nu-m_0}_2^2
        &=
        \E\norm{m_{\nu \mid x_i}-m_\nu}_2^2 \\
        &=
        \E_{x_i}\left[
        \E_{x \sim \nu \mid x_i}\norm{x-m_\nu}_2^2
        -
        \E_{x \sim \nu \mid x_i}\norm{x-m_{\nu \mid x_i}}_2^2
        \right] \\
        &=
        \Tr(\Sigma_\nu)-\E\Tr(\Sigma_{\nu \mid x_i})\,.
    \end{align*}
    The second equality uses $x-m_\nu=(x-m_{\nu \mid x_i})+(m_{\nu \mid x_i}-m_\nu)$ and the fact that $\E[x-m_{\nu \mid x_i}\mid x_i]=0$. Averaging over the sampled coordinate $i$ and over $\nu \sim \pi_\ell$ proves the claim.
\end{proof}

Note that~\cref{claim:fine-grained-step} gives an unconditional lower bound on the increase in our potential function $s_\ell$.
We now translate this to a concrete lower bound on the increase whenever the mean rounding success condition fails.

\begin{claim}
\label{claim:fine-grained-recurrence}
    If
    \[
        \E_{\nu \sim \pi_\ell}\Norm{\Sigma_\nu}_F
        >
        \gamma \E_{\nu \sim \pi_\ell}\norm{m_\nu-m_0}_2^2,
    \]
    then
    \[
        s_{\ell+1}-s_\ell
        \geq
        \max\left\{
            \frac{1-s_\ell}{n},
            \frac{\gamma^2s_\ell^2}{1-s_\ell}
        \right\}.
    \]
\end{claim}
\begin{proof}
    First, $\Norm{\Sigma_\nu}_F^2/\Tr(\Sigma_\nu)\geq \Tr(\Sigma_\nu)/n$ for every branch, so \Cref{claim:fine-grained-step} gives
    \[
        s_{\ell+1}-s_\ell\geq \frac{1-s_\ell}{n}.
    \]
    Second, by Cauchy-Schwarz and the identities in \Cref{claim:fine-grained-step},
    \begin{align*}
        \frac{1}{\Tr(\Cov(\mu_0))}\E_{\nu \sim \pi_\ell}
        \frac{\Norm{\Sigma_\nu}_F^2}{\Tr(\Sigma_\nu)}
        &\geq
        \frac{1}{\Tr(\Cov(\mu_0))}\cdot
        \frac{\left(\E_{\nu \sim \pi_\ell}\Norm{\Sigma_\nu}_F\right)^2}
        {\E_{\nu \sim \pi_\ell}\Tr(\Sigma_\nu)}  \\
        &\ge
        \frac{\gamma^2s_\ell^2}{1-s_\ell}.
    \end{align*}
    Combining the two bounds proves the claim.
\end{proof}

We now first use this to show as a warmup that $O\left(\frac{\sqrt{n}}{\gamma}\log n\right)$ steps of conditioning suffices to fulfill the mean rounding success condition.
This follows via a simple analysis of the recurrence outlined in the previous claim.

\begin{claim}
\label{claim:fine-grained-simple-hitting}
    Let $0<\gamma\leq \sqrt n$. Suppose $s_0=0$ and, until reaching $s_{\mathrm{end}}=1/\gamma^2$, a sequence satisfies the recurrence in Claim~\ref{claim:fine-grained-recurrence}. Then $s_\ell \geq s_{\mathrm{end}}$ for some
    \[
        \ell \leq O\left(\frac{\sqrt n}{\gamma}\log n\right).
    \]
\end{claim}
\begin{proof}
    The first term of the recurrence gives $s_1\geq 1/n$. Until the endpoint is reached, the recurrence implies
    \[
        s_{\ell+1}-s_\ell
        \geq
        \sqrt{\frac{1-s_\ell}{n}\cdot \frac{\gamma^2s_\ell^2}{1-s_\ell}}
        =
        \frac{\gamma}{\sqrt n}s_\ell.
    \]
    Hence $s_{\ell+1}\geq (1+\gamma/\sqrt n)s_\ell$ until the endpoint is reached. Starting from $s_1\geq 1/n$, this takes at most $O((\sqrt n/\gamma)\log n)$ additional rounds.
\end{proof}

We now give a slightly refined analysis of the recurrence to show that in fact $O\left(\frac{\sqrt{n}}{\gamma}\right)$ steps of conditioning suffices.
\begin{claim}
\label{claim:fine-grained-hitting}
    Let $1/\sqrt n \leq \gamma \leq \sqrt n$. Suppose $s_0=0$ and, until reaching $s_{\mathrm{end}}=1/\gamma^2$, a sequence satisfies the recurrence in Claim~\ref{claim:fine-grained-recurrence}. Then $s_\ell\geq s_{\mathrm{end}}$ for some
    \[
        \ell \leq O\left(\frac{\sqrt n}{\gamma}\right).
    \]
\end{claim}
\begin{proof}
    First we use the first term in the recurrence to reach $s_\ell \geq 1/(2\gamma\sqrt n)$. Since $\gamma \geq 1/\sqrt n$, this threshold is at most $1/2$, and
    \[
        1-s_{\ell+1}\leq \left(1-\frac1n\right)(1-s_\ell).
    \]
    Hence this first regime takes at most
    \[
        O\left(n\log\left(\frac{1}{1-1/(2\gamma\sqrt n)}\right)\right)
        \leq
        O\left(\frac{\sqrt n}{\gamma}\right)
    \]
    rounds.
    Next use the simpler consequence
    \[
        s_{\ell+1}\geq s_\ell(1+\gamma^2s_\ell).
    \]
    Consider a bucket $s_\ell \in [a,2a]$ with $a \leq 1/\gamma^2$. While the sequence remains in this bucket,
    \[
        s_{\ell+1}\geq s_\ell(1+\gamma^2a),
    \]
    so the bucket is crossed in $O(1/(\gamma^2a))$ rounds. Summing over geometrically increasing buckets starting at $1/(2\gamma\sqrt n)$, this second regime reaches either $s_{\mathrm{end}}$ within
    \[
        O\left(\frac{1}{\gamma^2\cdot 1/(2\gamma\sqrt n)}\right)
        =
        O\left(\frac{\sqrt n}{\gamma}\right)
    \]
    additional rounds, concluding the proof.
\end{proof}

We now combine the previous claims to prove the pinning lemma.
\begin{proof}[Proof of \cref{lem:fine-grained-pinning-loglog}]
    If $\Tr(\Cov(\mu_0))=0$, then $\mu_0$ is supported on a point and the conclusion holds without any conditioning. If $\gamma<1/\sqrt n$, condition on all coordinates; this uses $n\leq \sqrt n/\gamma$ conditionings and leaves zero covariance. Otherwise let $\pi_0$ be the point mass on $\mu_0$.

    If
    \[
        \E_{\nu \sim \pi_\ell}\Norm{\Sigma_\nu}_F
        >
        \gamma \E_{\nu \sim \pi_\ell}\norm{m_\nu-m_0}_2^2
    \]
    for every $\ell < T = O\left(\frac{\sqrt n}{\gamma}\right)$,
    then \Cref{claim:fine-grained-recurrence,claim:fine-grained-hitting} imply that some $s_\ell$ reaches $s_{\mathrm{end}}=1/\gamma^2$ by time $T$, which corresponds to the second term in the minimum, completing the proof. It is easy to see that the entire proof also goes through for pseudodistributions of appropriately high degree, since we only use the Cauchy--Schwarz inequality and conditioning.
\end{proof}

\printbibliography

\end{document}